\documentclass[acmsmall,screen,nonacm]{acmart}
\AtBeginDocument{%
  }

\usepackage{amsmath, amsfonts, amsthm}
\usepackage{mathpartir} 
\usepackage{stmaryrd}   
\usepackage[dvipsnames]{xcolor}
\usepackage{xspace}
\usepackage{subfiles}
\usepackage{cleveref}
\usepackage{pifont} 
\usepackage{listings}
\usepackage[frozencache]{minted} 
\usepackage{tikz} 
\usepackage{wrapfig}
\usepackage{enumitem}
\usepackage{float}

\usetikzlibrary{calc}

\setminted{
    mathescape,
    fontsize=\footnotesize,
    linenos,
    stepnumber=1,
    numbersep=5pt,
    xleftmargin=1.5em,
    tabsize=4,
}

\newcommand{\name}{MPFX\xspace}

\newcommand{\softfl}{SoftFloat\xspace}

\newcommand{\Z}{\mathbb{Z}\xspace}
\newcommand{\R}{\mathbb{R}\xspace}

\newcommand{\F}{{\mathbb{F}\xspace}}

\newcommand{\rnd}{{\text{rnd}\xspace}}

\newcommand{\A}{{\mathcal{A}\xspace}}

\newcommand{\ImplSLOC}{2,318\xspace}

\newcommand{\EvalOS}{Ubuntu 26.04 LTS\xspace}
\newcommand{\EvalMem}{512\xspace}
\newcommand{\EvalCPU}{AMD EPYC 7702P\xspace}
\newcommand{\EvalCC}{GCC 15.2.0\xspace}
\newcommand{\EvalMPFR}{4.2.2\xspace}
\newcommand{\EvalSoftfl}{Release 3e\xspace}
\newcommand{\EvalCPFloat}{v0.6.0\xspace}

\newcommand{\EvalCmpNumInputs}{10 million\xspace}
\newcommand{\EvalCmpNumIters}{30\xspace}
\newcommand{\EvalMaxSpeedup}{1.62\times\xspace}
\newcommand{\EvalMinSpeedup}{0.44\times\xspace}

\newcommand{\CaseStudyMismatches}{5.2\%\xspace}

\newcommand{\SpeedupOverMPFR}{{13.5\times\xspace}}

\newcommand{\SpeedupMeanOverMPFR}{{6.40\times\xspace}}
\newcommand{\SpeedupMeanOverSoftfloat}{{0.84\times\xspace}}
\newcommand{\SpeedupOverSoftfloatCaseStudy}{{13.3\times\xspace}}

\newtheorem{lemma}{Lemma}[section]

\crefname{lemma}{lemma}{lemmas}
\Crefname{lemma}{Lemma}{Lemmas}

\lstdefinestyle{CStyle}{
    language=C,
    basicstyle=\ttfamily\footnotesize,  
    keywordstyle=\color{blue},       
    commentstyle=\color{gray},       
    stringstyle=\color{red},         
    numbers=left,                    
    numberstyle=\tiny\color{gray},   
    stepnumber=1,                    
    showspaces=false,                
    showstringspaces=false,          
    showtabs=false,                  
    tabsize=4,                       
}

\lstdefinestyle{Pseudo}{
    language=Python,
    basicstyle=\ttfamily\footnotesize,  
    keywordstyle=\color{blue},       
    commentstyle=\color{gray},       
    stringstyle=\color{red},         
    numbers=left,                    
    numberstyle=\tiny\color{gray},   
    stepnumber=1,                    
    showspaces=false,                
    showstringspaces=false,          
    showtabs=false,                  
    tabsize=4,                       
    mathescape,                      
}

\newcommand{\hwop}[1]{{\textcolor{BurntOrange}{#1}}}   
\newcommand{\nlop}[1]{{\textcolor{Mulberry}{#1}}}      

\ifdefined\hideappendix
  \newcommand{\refRoundContain}{Section 5.1\xspace}
  \newcommand{\refRoundSafe}{Section 5.2\xspace}
  \newcommand{\refLemContainComplete}{Lemma 5.1\xspace}
  \newcommand{\refLemKeptDigits}{Lemma 5.2\xspace}
  \newcommand{\refLemKeptCmp}{Lemma 5.3\xspace}
  \newcommand{\refLemRtoRounding}{Lemma 5.4\xspace}
  \newcommand{\refTabDoubleRoundCex}{Table 2\xspace}
  \newcommand{\refAnalysisInfer}{Section 6.1\xspace}
  \newcommand{\refFigCorrectDoubleRound}{Figure 8\xspace}
\else
  \newcommand{\refRoundContain}{\Cref{ssec:round:contain}\xspace}
  \newcommand{\refRoundSafe}{\Cref{ssec:round:safe}\xspace}
  \newcommand{\refLemContainComplete}{\Cref{lem:contain-complete}\xspace}
  \newcommand{\refLemKeptDigits}{\Cref{lem:kept-digits}\xspace}
  \newcommand{\refLemKeptCmp}{\Cref{lem:kept-cmp}\xspace}
  \newcommand{\refLemRtoRounding}{\Cref{lem:rto-rounding}\xspace}
  \newcommand{\refTabDoubleRoundCex}{\Cref{fig:double-round-cex}\xspace}
  \newcommand{\refAnalysisInfer}{\Cref{ssec:analysis:infer}\xspace}
  \newcommand{\refFigCorrectDoubleRound}{\Cref{fig:correct-double-round}\xspace}
\fi

\begin{document}

\title{When Double Rounding is Correct}

\author{Brett Saiki}
\email{bsaiki@cs.uw.edu}
\orcid{0009-0002-3482-5767}
\affiliation{%
  \institution{University of Washington and Intel}
  \city{Seattle}
  \state{Washington}
  \country{USA}
}

\author{Bill Zorn}
\affiliation{%
  \institution{Intel}
  \city{Folsom}
  \state{CA}
  \country{USA}
}
\email{bill.zorn@intel.com}

\author{Cynthia Richey}
\email{lapwing@seas.upenn.edu}
\orcid{0009-0008-4456-9406}
\affiliation{%
  \institution{University of Pennsylvania}
  \city{Philadelphia}
  \state{PA}
  \country{USA}
}

\author{Zachary Tatlock}
\affiliation{%
  \institution{University of Washington}
  \city{Seattle}
  \state{WA}
  \country{USA}
}
\orcid{0000-0002-4731-0124}
\email{ztatlock@cs.uw.edu}

\renewcommand{\shortauthors}{Saiki et al.}

\begin{abstract}


Number libraries, which implement number formats and their arithmetic in software, are
  used to simulate hardware, making hardware behavior reproducible and thus
  verifiable and easier to design.
However, modern machine learning accelerators are difficult to simulate, as they increasingly
  use specialized number formats that existing number libraries do not support.
It is costly to extend these libraries using traditional methods:
  each combination of operation, format, and rounding mode typically requires a 
  bespoke implementation.
An easier way is to 
\emph{double round}---%
  compute at higher precision, then re-round---%
  so that one high-precision kernel serves many formats, but 
  double rounding is only known to be correct in select cases.

We give sufficient conditions for when double rounding
  is correct across many formats and rounding modes.
Our result rests on an abstract number format that parameterizes
  fixed- and floating-point formats,
  reducing correct double rounding to \emph{format containment}---%
  whether every value of one format is representable in another---%
  for which we give an efficient, complete check on the formats' parameters.
These results are mechanized in Lean~4.
In addition,
  we present a methodology, built on \emph{format inference},
  for determining when roundings can be dropped, inserted,
  or modified without changing a program's result.

We apply these insights in two ways.
First,
  we present \name, a correctly-rounded, multi-precision number library
  that exploits correct double rounding to efficiently simulate
  many number formats.
\name's correctly-rounded operations
  are up to $\SpeedupOverMPFR$ (mean: $\SpeedupMeanOverMPFR$)
  faster than MPFR's,
  and competitive with SoftFloat's---between $\EvalMinSpeedup$ and
  $\EvalMaxSpeedup$ as fast (mean: $\SpeedupMeanOverSoftfloat$)---%
  at the same precision.
Second,
  in a case study,
  we implement a software simulation of a hardware specification
  using both SoftFloat and \name's primitives;
  the \name version is up to $\SpeedupOverSoftfloatCaseStudy$ faster.

\end{abstract}

\maketitle

\section{Introduction}
\label{sec:intro}

The growing importance of machine learning workloads
  has driven a rapid proliferation of new \textit{number formats}~\cite{bfloat16, ocp-mx, p3109, graphcore-fp8},
  including 8-, 6-, and even 4-bit floating-point formats, which 
  improve performance and energy efficiency at the cost of accuracy.
Fixed-function hardware units (e.g., tensor cores) consume these
  small formats and perform sequences of operations on them internally,
  using exact or high-precision arithmetic for intermediate computations.
A developer exploring the space of new hardware designs (or
  reproducing or verifying existing ones) faces a challenge:
  general-purpose CPUs do not support many of these formats,
  requiring a software simulation for each new candidate design
  that reproduces its behavior exactly.

One way to emulate unusual number formats is to use a 
  \emph{number library}~\cite{softfloat, softposit, mpfr, universal, floppyfloat}. 
A number library simulates number formats using the available numerical capabilities 
  of the programming environment---for example, the Berkeley SoftFloat library~\cite{softfloat}
  implements IEEE 754 floating-point arithmetic using only machine integer operations. 
As trusted reference implementations, 
  number libraries are the backbone of tools including
  floating-point expression rewriters~\cite{herbie, chassis},
  math library generators~\cite{metalibm,megalibm,sollya,rlibm},
  and error analyzers~\cite{gappa,vcfloat,vcfloat2,fptaylor}.

However,
  most number libraries support operations over a limited set of number formats,
  and extending these operations to new formats requires substantial expertise,
  careful design, and extensive validation~\cite{correct-round}.
Since a format is tightly coupled
  with its arithmetic and rounding behavior,
  a correct implementation of one format does not transfer to another:
  developers must construct a new implementation
  for \textit{every combination} of format, operation,
  and rounding mode.

\textit{Double rounding} is an alternative approach to building number libraries
  which relies on the high-precision arithmetic available on most modern
  processors (e.g., IEEE 754 double-precision floating-point).
The idea is simple:
  to simulate a low-precision format,
  compute the high-precision result,
  then re-round to the target format.
By \emph{separating arithmetic from rounding},
  a single high-precision kernel efficiently serves many number formats.
But this works only when double rounding is \textit{correct}---that is,
  when it yields the same result as rounding directly
  to the low-precision format.
Prior work establishes this
  for particular pairs of rounding modes~\cite{boldo-rto, flocq, rlibm-all}
  or particular operations~\cite{figueroa, rump-double-round, roux-double-round, cpfloat},
  but the question remains: when is double rounding correct \textit{in general}?

To answer this question, 
  we introduce an \textit{abstract number format} that
  captures both floating-point and fixed-point number formats;
  its parameterization permits an efficient, complete check for determining
  \emph{format containment}---when
  every representable value in one number format
  is representable in another format.
We show that correct double rounding
  reduces to a format containment check:
  it suffices that the intermediate format contain the target format,
  possibly with a few extra bits of precision
  (depending on the rounding modes involved). 
Our results extend prior work~\cite{boldo-rto, flocq, rlibm-all}
  to more rounding modes, to formats with a maximum value,
  and to special values and overflow.
We mechanize all of these results in Lean~4~\cite{lean4}.

Beyond guaranteeing that double rounding is correct for individual operations,
  we show how our abstract number format lets us map a whole real-valued
  specification onto the finite operations available in a programming environment.
In practice,
  machine learning accelerators rarely operate
  on small formats end-to-end; instead,
  these formats serve as inputs to longer computations,
  with intermediate steps using exact or high-precision arithmetic.
A dot product over 8-bit OCP MX formats~\cite{ocp-mx}, for example,
  accumulates into an intermediate value requiring
  far more precision than its small inputs would suggest.
We develop a methodology built on \textit{format inference},
  which bounds each intermediate computation by a format
  covering its full range of values;
  with format containment and correct double rounding,
  it justifies semantics-preserving rewrites that map
  unrounded or high-precision operations
  onto those available in the programming environment.

We apply our theory in two ways.
First, we build \name,
  a \emph{correctly-rounded}, \emph{multi-precision}
  number library in C++ that is correct for any instance of our
  abstract number format whose required intermediate (\Cref{sec:round})
  fits in the IEEE 754 double-precision floating-point format.
It handles every format by correctly re-rounding a high-precision result,
  rather than through the per-format bespoke implementations
  that number libraries traditionally require.
\name builds on four arithmetic backends---hardware
  floating-point, two existing number libraries, and custom
  kernels based on error-free transformations (EFTs)~\cite{eft}.
It is efficient:
  its correctly-rounded operations run up to
  $\SpeedupOverMPFR$ (mean: $\SpeedupMeanOverMPFR$) faster than MPFR,
  and are competitive with SoftFloat---between $\EvalMinSpeedup$ and
  $\EvalMaxSpeedup$ as fast (mean: $\SpeedupMeanOverSoftfloat$).
We compare to SoftFloat only on the formats SoftFloat is highly optimized 
  for (e.g., single-precision floating point);
  \name remains competitive on those formats while supporting many more.
Second, in a case study, we use format inference and correct double
  rounding to implement an algorithm from the OCP MX standard~\cite{ocp-mx},
  first with SoftFloat and then with \name;
  our \name-based program is up to $\SpeedupOverSoftfloatCaseStudy$ faster
  than the SoftFloat version.

Concretely, this paper makes the following contributions:
\begin{itemize}
\item We present an \emph{abstract number format} (\Cref{sec:format})
  that parameterizes fixed- and floating-point number formats
  and permits an efficient and complete check for
  \emph{format containment} (\Cref{ssec:round:contain}).
\item We characterize when \emph{double rounding is correct},
  expressing sufficient conditions as format containment,
  and when it is incorrect no matter how wide the intermediate format
  (\Cref{ssec:round:safe}).
\item We present a methodology, built on \emph{format inference},
  for determining when dropping, inserting, or modifying
  a rounding is \emph{semantics-preserving}
  (\Cref{sec:analysis}).
\end{itemize}

Subsequently,
  \Cref{sec:eval} introduces \name and evaluates its performance,
  and \Cref{sec:case_study} presents our case study.
\Cref{sec:future} describes limitations and future work;
  \Cref{sec:related} discusses related work;
  \Cref{sec:conclusion} summarizes our conclusions.
Our Lean mechanization is included as supplementary material.
We transcribed the paper's definitions and theorem statements into Lean by hand,
  then used large language models (Claude Opus~4.7 and~4.8)
  to mechanize the prose proofs of Appendix~A.

\section{Overview}
\label{sec:overview}

The OCP MX standard~\cite{ocp-mx} defines small number formats, as well as
  algorithms for efficient training and inference.
These algorithms describe computations at a high level,
  but leave specific implementation details (e.g., the number
  formats and rounding modes involved) to the implementer.

\Cref{fig:overview:mx} outlines an MX algorithm for 
  computing a dot product over \textit{blocks}, which 
  is the fundamental unit of the standard---essentially 
  a vector $\{ V_i \}$ of numbers sharing a common 
  \textit{scale factor}, \textit{X}, where each
  number in the block is represented as $V_i = X \times P_i$.
The dataflow diagram (right), which represents one iteration of the loop, 
  is annotated with question
  marks where concrete choices of number format and rounding
  mode must be made before the behavior of the algorithm on hardware
  is fully specified.
Each way of resolving these choices yields a different hardware design,
  trading numerical accuracy for cost in chip area and power\footnote{
    Selecting these formats and precisions is
      \emph{precision tuning}~\cite{precimonious, fptuner, daisy, popix, isychev-lohar};
      we instead ask how to simulate even
      a \emph{single} candidate selection
      faithfully and efficiently.
  }.

\begin{figure}[ht!]
\noindent
\begin{minipage}[t]{0.49\textwidth}
\vspace{0pt}
\begin{minted}[escapeinside=||]{python}
# dot product over MX blocks
def mx_dot(|$A$|, |$B$|):
  |$C$| = 0
  for (|$X^{A}_{i}$|, |$P^{A}_{i}$|), (|$X^{B}_{i}$|, |$P^{B}_{i}$|) in zip(|$A$|, |$B$|):
    |$C$| += |$X^{A}_{i} \times X^{B}_{i} \times \sum_{j=1}^{32}\left( P^{A}_{i,j} \times P^{B}_{i,j}\right)$|
  return |$C$|
\end{minted}
\end{minipage}%
\hfill
\begin{minipage}[t]{0.49\textwidth}
\centering
\vspace{0pt}
\includegraphics[width=\textwidth]{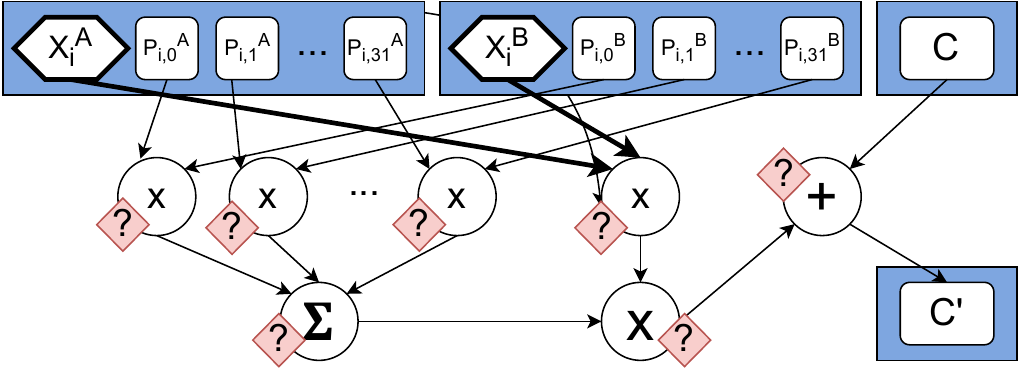}
\end{minipage}
\caption{%
A dot product over MX blocks (left)
  and the dataflow diagram for a single loop iteration (right),
  which combines one block from $A$ and one from $B$ with the running
  accumulator $C$ to produce the updated $C$.
Both $A$ and $B$ are encoded as MX blocks
  where each block has a scale factor $X_i$
  and elements $\{ P_{i,j} \}$.
}
\label{fig:overview:mx}
\end{figure}

Consider an engineer modeling candidate accelerators in software.
Before committing a design to silicon,
  they must check that their design choices yield acceptable numerical behavior when
  running training or inference.
\Cref{fig:overview:mx_spec} (left) describes a candidate
  that performs as few roundings as possible.
It uses a single input format---say, the 8-bit
  MX format E5M2---and performs all intermediate computations
  \textit{exactly}, rounding only at the final accumulation.


\begin{figure}[ht!]
\noindent
\begin{minipage}[t]{0.44\textwidth}
\begin{minted}[escapeinside=||]{Python}
def accurate_mx_dot(|$A$|, |$B$|):
  |$C$| = 0
  for (|$X^{A}_i$|, |$P^{A}_i$|), (|$X^{B}_i$|, |$P^{B}_i$|) in zip(|$A$|, |$B$|):
    |$X$| = |$\text{mul}_{\R,\_}(X^{A}_i, X^{B}_i)$|
    |$V$| = 0
    for |$P^{A}_{i,j}$|, |$P^{B}_{i,j}$| in zip(|$P^{A}_i$|, |$P^{B}_i$|):
      |$t_1$| = |$\text{mul}_{\R,\_}(P^{A}_{i,j}, P^{B}_{i,j})$|
      |$V$| = |$\text{add}_{\R,\_}(V, t_1)$|
   |$t_2$| = |$\text{mul}_{\R,\_}(X, V)$|
   |$C$| = |$\text{add}_{\text{FP32},\text{RNE}}(C, t_2)$|
  return |$C$|
\end{minted}
\end{minipage}%
\hfill
\begin{minipage}[t]{0.54\textwidth}
\begin{minted}[escapeinside=||]{Python}
def our_mx_dot_mpfx(|$A$|, |$B$|):
  |$C$| = 0
  for (|$X^{A}_i$|, |$P^{A}_i$|), (|$X^{B}_i$|, |$P^{B}_i$|) in zip(|$A$|, |$B$|):
    |$X$| = |$\hwop{\text{mul}_{\text{FP64},\text{RNE}}}(\hwop{\text{rnd}_{\text{FP64},\text{RNE}}}(X^{A}_i), \hwop{\text{rnd}_{\text{FP64},\text{RNE}}}(X^{B}_i))$|
    |$V$| = 0
    for |$P^{A}_{i,j}$|, |$P^{B}_{i,j}$| in zip(|$P^{A}_i$|, |$P^{B}_i$|):
      |$t$| = |$\hwop{\text{mul}_{\text{FP32},\text{RNE}}}(\hwop{\text{rnd}_{\text{FP32},\text{RNE}}}(P^{A}_{i,j}), \hwop{\text{rnd}_{\text{FP32},\text{RNE}}}(P^{B}_{i,j}))$|
      |$V$| = |$\hwop{\text{add}_{\A(\infty,-32,2^{95} - 2^{-32}),\text{RNE}}}(V, \nlop{\text{rnd}_{\A(\infty,-32,2^{95} - 2^{-32}),\text{RNE}}}(t))$|
    |$V_{hi}, V_{lo}$| = |$\nlop{\text{split}}(V, 40)$|
    |$t_2$| = |$\hwop{\text{mul}_{\text{FP64},\text{RNE}}}(X, \nlop{\text{rnd}_{\text{FP64},\text{RNE}}}(V_{hi}))$|
    |$t_3$| = |$\hwop{\text{mul}_{\text{FP64},\text{RNE}}}(X, \nlop{\text{rnd}_{\text{FP64},\text{RNE}}}(V_{lo}))$|
    |$C$| = |$\nlop{\text{add3}_{\text{FP32},\text{RNE}}}(t_2, t_3, C)$|
  return |$C$|
\end{minted}
\end{minipage}
\caption{%
The MX dot product:
  a fully specified algorithm (left);
  and its efficient \name implementation (right),
  using \nlop{\name} or \hwop{hardware-supported} operations,
  derived in \Cref{sec:case_study}.
Each operator $f_{F,rm}$ is annotated with
  its number format $F$ and rounding mode $rm$
  (an unrounded operation is denoted $f_{\R, \_}$).
}
\label{fig:overview:mx_spec}
\end{figure}


To model this candidate faithfully in software,
  the engineer faces an immediate obstacle:
  its operations do not exist in the programming environment.
The specification computes with infinitely-precise operations,
  whereas the environment provides only
  rounded operations over a handful of formats.
Faithfully simulating the design therefore means rewriting
  the program to replace each of its operations with
  one or more available ones.
This requires knowing what format each intermediate value inhabits,
  and when such a replacement is \emph{semantics-preserving}.
\Cref{sec:analysis} shows how to accomplish this
  using \emph{format inference} and
  correct double rounding (\Cref{sec:round}).



As it turns out, this candidate is impractical:
  despite its small inputs,
  the inner accumulator of the dot product in \Cref{fig:overview:mx_spec}
  (left, line 8) requires 69 bits of precision---well beyond double precision (53 bits).
A hardware architect is unlikely to build such an accumulator;
  they would accept some rounding error and narrow it to a more manageable
  width---say, 20 to 50 bits.
But the accumulator is only one of the operations
  left unspecified by \Cref{fig:overview:mx}:
  each has a number format and rounding mode to choose,
  and each combination is its own candidate.
Today, exploring this space means constructing a faithful model
  for each candidate by hand, reasoning about its formats,
  and validating it individually.
Established number libraries offer little relief:
  each supports only a fixed handful of formats,
  so any candidate outside them requires another bespoke implementation.


Our number library, \name, is built with flexibility in mind:
  it is correct for any instance of our abstract number format (\Cref{sec:format})
  whose required intermediate fits in double precision,
  with performance comparable to efficient number libraries
  like SoftFloat~\cite{softfloat} and FloppyFloat~\cite{floppyfloat}.
Rather than implementing each candidate from scratch,
  \name exploits cases where double rounding is correct,
  reusing a few high-precision kernels across many formats by re-rounding their results;
  we discuss the relevant theory in \Cref{sec:round}.

In \Cref{sec:case_study},
  we map the specification in
  \Cref{fig:overview:mx_spec} (left) to SoftFloat's primitives,
  and show that targeting \name yields a more efficient
  implementation (right).
Crucially,
  both implementations are \emph{semantically equivalent} to the specification,
  yet they use only the operations the programming environment provides.
The rest of this paper addresses the following central problem:
  \emph{Given a numerical specification,
  how do we construct a faithful yet efficient implementation?}



\section{What is a Number Library?}
\label{sec:bg}

A \textit{number library} simulates number formats---%
  floating-point, fixed-point, posits, and more---%
  beyond those offered by standard hardware or language runtimes.
Number libraries are essential tools for analyzing numerical error in 
  multi-precision numerical algorithms, exploring different 
  implementation trade-offs, and verifying the correctness of 
  numerical software and hardware.
Despite the diversity of number formats they support,
  number libraries share common ideas:
  \textit{number formats} define which values are representable (\Cref{ssec:bg:formats}),
  and \textit{rounding} specifies how to approximate unrepresentable values (\Cref{ssec:bg:rounding}).
Ideally,
  rounding satisfies a correctness property (\Cref{ssec:bg:correct-round}),
  which extends to compositions of roundings,
  or \textit{double roundings} (\Cref{ssec:bg:double-round}).
We close the section with a survey of
  existing number libraries (\Cref{ssec:bg:survey}).

\subsection{Number Formats}
\label{ssec:bg:formats}

A \emph{number format} defines
  a discrete set of values that approximate the real numbers. 
Examples of number formats include
  machine integers (e.g., 32-bit signed integers);
  fixed-point numbers;
  floating-point numbers~\cite{ieee754, ocp-mx, p3109}
  (e.g., IEEE 754 single-precision floating-point numbers, OCP MX floating-point numbers);
  posit numbers~\cite{posit};
  and logarithmic numbers~\cite{lns}.
In practice,
  number formats are accompanied by
  a finite \emph{encoding}  that maps
  between bit patterns and numerical values
  and possibly non-real values such as
  infinities or NaNs.

In this paper,
  we restrict our attention to \emph{dyadic numbers},
  or numbers of the form
  ${(-1)}^{s} \times c \times 2^{exp}$ where
  $s \in \{ 0, 1 \}$ is the \emph{sign},
  $c \in \Z_{\geq 0}$ is the \emph{significand},
  and $exp \in \Z$ is the \emph{exponent}.
The set of dyadic numbers is denoted $\Z\left[\frac{1}{2}\right]$.
We will consider number formats
  that are subsets of the dyadic numbers\footnote{%
  The restriction to dyadic numbers excludes
    more esoteric formats like logarithmic number systems.};
  this covers the common number formats:
  floating-point, fixed-point, and integer number formats.
A number format may also contain \emph{special values},
  which indicate an extremal ($\pm\infty$) or undefined (NaN) result.

We denote the set of all possible number formats with $\F$.
Notice that our representation of
  dyadic numbers is redundant since
  $c \times 2^{exp} = 2c \times 2^{exp - 1}$; that is, 
  $3 \times 2^2 = 6 \times 2^1 = 12 \times 2^0$ are all valid
  representations of the dyadic number 12. Let $F \in \F$ be a number format.
A dyadic number $x \in \Z\left[\frac{1}{2}\right]$
 is \emph{representable} in $F$ when $x \in F$.
If $x = {(-1)}^s \times c \times 2^{exp}$,
  the precision of $x$ is $\lceil \log_2(c + 1) \rceil$,
  which is the fewest binary digits that can encode $c$.
For $y_1, y_2 \in F$,
  we say that $y_1$ and $y_2$ are \emph{adjacent}
  whenever there exists no intervening representable
  number.

\subsection{Rounding}
\label{ssec:bg:rounding}

For any number format,
  there is a gap between any two adjacent representable values,
  a range of real numbers that are unrepresentable in the format.
If we want to convert a real number
  in this gap to a number in the target format,
  we must \textit{round} it to either the lesser
  or greater value.
The choice between these options
  is governed by the \textit{rounding mode}.

A rounding function $\rnd_{F, rm}$
  maps every real number $x \in \R$ to
  a representable value in
  the number format $F \in \F$.
If $x$ is representable in $F$,
  then $\rnd_{F, rm}(x) = x$.
Otherwise,
  there exist two adjacent representable values
  $y_1, y_2 \in F$ such that $y_1 < x < y_2$\footnote{%
  So far, we assume $F$ has no maximum representable value
    or special values, so $y_1$ and $y_2$ always exist.};
  the rounding mode $rm$ decides whether
  $\rnd_{F, rm}(x) = y_1$ or $\rnd_{F, rm}(x) = y_2$.
We classify representable values as \emph{even} or \emph{odd}
  so that, for two adjacent representable values $y_1$ and $y_2$,
  one is always even and the other odd\footnote{%
  We usually take the parity of the integer significand
    in a canonical, format-dependent representation;
    when $p = 1$, however, every significand is $c = 1$,
    so we assign parity directly to keep values alternating, with $0$ even.
  This definition agrees with the treatment of oddness for
    1-bit formats in the IEEE P3109 draft standard~\cite{p3109}.}.



The IEEE 754 standard~\cite{ieee754} defines five rounding modes:
  two that round to the \textit{nearest} value (and differ in their tiebreaking behavior),
  and three that round in a fixed direction.
The nearest rounding modes choose
  the representable value that is closest to $x$.
When $x$ is exactly halfway between $y_1$ and $y_2$,
  a tiebreaking rule chooses which representable value
  to produce:
\begin{itemize}
\item \emph{round to nearest, ties to even} (RNE):
  tiebreak to the even value;
\item \emph{round to nearest, ties away from zero} (RNA):
  tiebreak to the value with the larger magnitude.
\end{itemize}
The directed rounding modes choose
  the representable value irrespective
  of distance:
\begin{itemize}
\item \emph{round towards positive infinity} (RTP):
  choose the value closest to $+\infty$;
\item \emph{round towards negative infinity} (RTN):
  choose the value closest to $-\infty$;
\item \emph{round towards zero} (RTZ):
  choose the value with the smaller magnitude.
\end{itemize}
For this paper,
  we consider two other rounding modes:
\begin{itemize}
\item \emph{round away from zero} (RAZ):
  choose the value with the larger magnitude;
\item \emph{round to odd} (RTO):
  choose the odd value.
\end{itemize}
The RTZ mode is the obvious counterpart to RAZ.
The RTO mode is a subject of much research interest
  in recent decades~\cite{boldo-rto, boldo-rto-apps, rlibm}.
For this paper,
  we will focus on RNE, RTZ, RAZ, and RTO
  since the specific tiebreaking rule of
  RNE and RNA has no effect in our analysis,
  and RTP acts as RAZ when $x \geq 0$ and as RTZ otherwise,
  with RTN the opposite.

In practice,
  number formats have a maximum representable value $b$.
Let $\hat{F}$ be $F$ with no maximum representable value or special values.
Given functions
  $\sigma : \{+\infty, -\infty, \mathrm{NaN}\} \to F$
  and $\omega : \{-1, +1\} \to F$,
  rounding to $F$ is
\[
\rnd^{\sigma, \omega}_{F, rm}(x) =
\begin{cases}
  \sigma(x) & \text{if } x \in \{+\infty, -\infty, \mathrm{NaN}\}, \\
  \omega(\mathrm{sgn}(y)) & \text{if } |y| > b, \\
  y & \text{otherwise},
\end{cases}
\qquad \text{where } y = \rnd_{\hat{F}, rm}(x).
\]
We say $x$ \emph{overflows} when $|y| > b$.
We omit $\sigma$ and $\omega$ when clear from context.
IEEE 754~\cite{ieee754} sets $\sigma$ to the identity
  and $\omega(s)$ to $s \cdot \infty$ or $s \cdot b$,
  depending on the rounding mode and $s$;
  OCP MX~\cite{ocp-mx-spec} may instead saturate,
  with $\sigma(\pm\infty) = \omega(\pm 1) = \pm b$.


\subsection{Correct Rounding}
\label{ssec:bg:correct-round}

So far, we have discussed rounding individual real numbers
  to a target number format.
However,
  rounding often occurs as part of a computation:
  when we add two floating-point numbers,
  their sum may not be representable.
We want a guarantee that numerical operations give
  the same result as performing the exact computation,
  \textit{then} rounding.

Let $f: \R^n \to \R$ be a function.
An implementation of $f$
  for a number format $F$ and rounding mode $rm$,
  denoted $f_{F, rm}$, is \emph{correctly-rounded}
  if it produces the rounded value of the exact result of $f$.
More formally,
  $f_{F, rm}$ is the composition of
  the exact function $f$ and a rounding
  function $\rnd_{F, rm}$:
\[
  f_{F, rm} = \rnd_{F, rm} \circ f.
\]

Few languages and libraries offer comprehensive correct rounding.
The IEEE 754 standard requires only that
  basic arithmetic (e.g., addition) be correctly rounded;
  in most languages,
  implementations of transcendental functions---%
  for example, the exponent function $e^x$---are not guaranteed
  to be correctly rounded.
This lack of clear specification means that
  the results of numerical algorithms may vary across operating systems
  and languages, leading to higher than expected error.
Requiring correct rounding in
  programming environments has strong benefits~\cite{correct-round}:
  it provides a clear specification for the behavior of $f_{F, rm}$;
  it ensures reproducible results across
  different implementations of $f_{F, rm}$;
  and it bounds the numerical error introduced
  by rounding consistently.

\subsection{Double Rounding}
\label{ssec:bg:double-round}


A \emph{double rounding} is the composition of two rounding functions:
  for number formats $F_1, F_2$ and rounding modes $rm_1, rm_2$,                                                                                                                        
  a double rounding has the form
  $\rnd_{F_1, rm_1} \circ \rnd_{F_2, rm_2}$.
A double rounding is \emph{correct}
  when the intermediate rounding has no effect on the final result:
\[
\rnd_{F_1, rm_1} = \rnd_{F_1, rm_1} \circ \rnd_{F_2, rm_2}.
\]
Correct double rounding lets us safely re-round
  an already-rounded computation.
Composing both sides of the identity with $f$ yields
\[
\rnd_{F_1, rm_1} \circ f = \rnd_{F_1, rm_1} \circ \rnd_{F_2, rm_2} \circ f,
\]
  that is, $f_{F_1, rm_1} = \rnd_{F_1, rm_1} \circ f_{F_2, rm_2}$:
  re-rounding the correctly-rounded result for $F_2$
  produces the correctly-rounded result for $F_1$,
  for \emph{any} function $f$.

Unfortunately,
  double roundings are not always correct roundings.
For example,
  let $x = 1.26$ be a rational number,
  $F_1$ be the E2M1 floating-point format from the OCP MX standard~\cite{ocp-mx}
  and $F_2$ be the E4M3 floating-point format from the same standard.
If $rm_1 = rm_2 = \text{RNE}$, then
\[
\rnd_{F_1, rm_1}(x) = 1.5,
\qquad
\rnd_{F_2, rm_2}(x) = 1.25,
\qquad\text{and}\qquad
\rnd_{F_1, rm_1}(\rnd_{F_2, rm_2}(x)) = 1,
\]
  so we have a clear counterexample.
If instead $rm_1 = rm_2 = \text{RTZ}$,
  both expressions evaluate to 1.


\subsection{A Brief Survey of Number Libraries}
\label{ssec:bg:survey}

Number libraries package software implementations of
  number formats, rounding modes, and the operations defined over them.
Each library must decide which formats and
  rounding modes to implement, as well as
  the correctness guarantees they provide.
We survey three \emph{correctly-rounded} libraries that make different
  choices across these dimensions.

The Berkeley SoftFloat~\cite{softfloat} library
  models IEEE 754 arithmetic using only
  machine integer operations.
It implements the required arithmetic
  for the four named IEEE 754 formats\footnote{%
  SoftFloat also supports 80-bit floating-point from
  the x86 architecture.}.

The FloppyFloat~\cite{floppyfloat} library
  simulates floating-point arithmetic using existing
  floating-point arithmetic for superior performance
  compared to SoftFloat.
It models other rounding modes using techniques such
  as error-free transformations (EFTs)~\cite{eft},
  requiring only the host system's floating-point
  operations under the default rounding mode (RNE).


Lastly,
  the MPFR~\cite{mpfr} library targets 
  numerical analysis applications, modeling a more
  mathematical form of floating-point arithmetic
  with no strict correspondence to hardware floating point.
The library supports \emph{arbitrary-precision} floating point,
  allowing the user to set the number of binary digits
  at the cost of requiring a memory allocator.
It provides a comprehensive set of correctly-rounded operations,
  including arithmetic, algebraic, and transcendental functions.

\section{Unifying Number Formats}
\label{sec:format}

In this section,
  we introduce an \emph{abstract number format}
  that encompasses both floating-point and fixed-point formats. 
Our key observation is that common number formats are not arbitrary
  sets of dyadic numbers; they exhibit significantly more structure. 
We exploit this structure by modeling
  a format's representable values as the output of parameterized 
  generators (\Cref{ssec:format:generators}), from which the 
  abstract number format follows naturally (\Cref{ssec:format:abstract}). 
In \Cref{sec:round}, 
  we will see how this abstraction allows us to efficiently 
  reason about format containment and correct double rounding.

\subsection{Generating Number Formats}
\label{ssec:format:generators}


To develop a unified abstract number format,
  we first construct \textit{generators}
  for fixed-point and floating-point formats
  that produce all the \emph{numerical} values in those formats,
  and then observe that these generators are similar.
To start,
  we will consider formats that have no maximum value
  and therefore contain countably infinite representable values,
  but we will later impose a bound.
We focus solely on the representable values themselves,
  not on their underlying encodings\footnote{%
  An encoding specifies a conversion between
  representable values and bit strings.
  }.

In \Cref{ssec:format:abstract},
  we discuss special values like infinity and NaN.
Our formats have a single zero;
  we do not model signed zero.






\subsubsection*{Fixed-Point Number Formats}

Fixed-point formats have
  representable values that are equidistant.
In particular,
  the distance between adjacent representable values,
  also called the \emph{quantum}, is a power
  of two\footnote{%
  The quantum could be a fixed multiple of a power
    of two, e.g., $\alpha \times 2^{exp}$ for some $\alpha \in \R$,
    but we set $\alpha = 1$ for simplicity.
  }: $q = 2^{exp}$.
The integers are a familiar instance
  with a quantum of $q = 1$.
The process for generating
  the fixed-point format with quantum $q$
  is straightforward.
First,
  we set $v = 0$ and produce 0.
Then,
  we repeatedly increase $v$ by $q$,
  producing $+v$ and $-v$ at each step.
\Cref{fig:format:fixed} shows
  the result of this procedure:
  an approximation of the real line by
  a set of evenly-spaced, dyadic numbers
  with spacing $q = 2^{-2}$.
The crucial property of this generator
  is that every value it produces
  has the form $\pm c \times 2^{exp}$.


\begin{figure}[ht!]
\includegraphics[width=\textwidth]{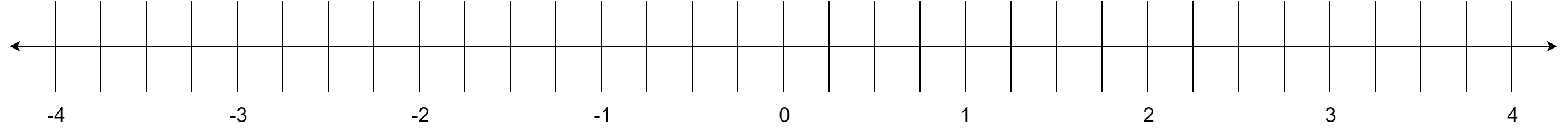}
\caption{%
Representable values on the real line for
  a fixed-point format with quantum $q = 2^{-2}$.
Fixed-point formats have representable
  values that are evenly spaced.
}
\label{fig:format:fixed}
\end{figure}

\subsubsection*{Floating-Point Number Formats}

A floating-point number format
  arises naturally by \emph{filtering} for values
  in a fixed-point format that satisfy
  a \emph{maximum precision} constraint:
  each representable value must have
  some form ${(-1)}^s \times c \times 2^{exp}$
  where $c$ has precision at most $p \in \Z_{\geq 1}$.

The procedure for generating representable floating-point values
  begins in an identical way to the fixed-point format;
  we repeatedly select values exactly
  $q_{\text{min}} = 2^{exp_{\text{min}}}$ apart,
  starting from 0.
Eventually,
  we encounter a number that, at first glance,
  seems to violate the precision constraint:
  $\pm 2^p \times 2^{exp_{\text{min}}}$ has
  a $(p + 1)$-digit significand.
Fortunately,
  we can take advantage of our redundant representation
  of dyadic numbers; we can represent the number with
  $\pm 2^{p-1} \times 2^{exp_{\text{min}}+1}$ instead,
  which has a $p$-digit significand.
However,
  its successor has no such luck:
  $\pm (2^p + 1) \times 2^{exp_{\text{min}}}$ can be represented
  with, at minimum, a $(p + 1)$-digit significand.
The generator filters out this value
  and then steps again with the quantum $q_{\text{min}}$.
Repeating this generate-and-test procedure
  constructs the representable values of
  the floating-point number format with initial quantum
  $q_{\text{min}}$ and maximum precision $p$.
\Cref{fig:format:float} shows the representable
  values of this number format.

\begin{figure}[ht!]
\includegraphics[width=\textwidth]{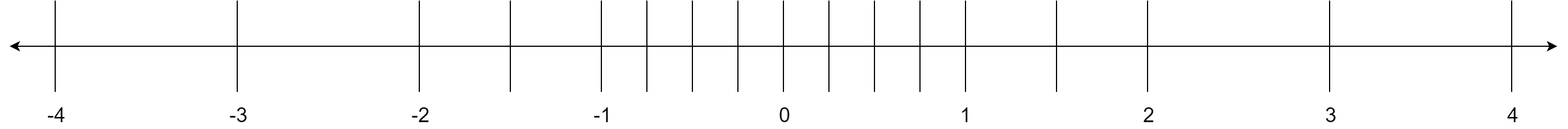}
\caption{%
Representable values on the real line for
  a floating-point format with an initial quantum
  $q = 2^{-2}$ and a maximum precision of $p = 2$.
}
\label{fig:format:float}
\end{figure}

At first,
  it may be surprising that this process generates
  the representable values of a floating-point format.
Examining the number line reveals
  key properties of floating-point numbers.
For $|x| \geq 1$,
  the spacing between representable values doubles
  after each power of two,
  and the number of representable values between
  powers of two $[2^e, 2^{e+1})$ is constant---%
  in this example, the number is 2.
For $|x| < 1$,
  the spacing is constant,
  a property called gradual underflow\footnote{%
  A nonzero representable value is \emph{subnormal} if
    $|x| < 2^{exp_{\text{min}} + p - 1}$ and \emph{normal} otherwise;
    in this example, $\pm 1/4$ is subnormal and every
    $|x| \geq 1/2$ is normal.
  }.
Importantly,
  since each floating-point format is a subset
  of some fixed-point format, the two number format families
  share the following property:
  \emph{every representable value is an integer multiple
  of the (initial) quantum}.


The aforementioned generator is inefficient;
  the acceptance rate halves whenever $v$
  crosses a power of two,
  eventually converging to zero.
We can make a better generator that
  produces values without failure.
Let
  $v_{\text{cut}} = 2^{p-1} \times 2^{exp_{\text{min}} + 1}
  = 2^{exp_{\text{min}} + p}$
  be the \emph{cutoff} after which the distance between
  representable values increases---%
  for example, $v_{\text{cut}} = 1$ in \Cref{fig:format:float}.
At $v_{\text{cut}}$,
  the generator doubles the quantum to $q = 2^{exp_{\text{min}} + 1}$
  and steps with the new quantum.
Each subsequent value is representable
  with the form $\pm c \times 2^{exp_{\text{min}} + 1}$
  until we encounter the next power of two,
  $2v_{\text{cut}} = \pm 2^{(p-1)} \times 2^{exp_{\text{min}}+2}$,
  at which point, the quantum doubles again.
This procedure continues indefinitely
  with the quantum doubling whenever
  the generator reaches the next power of two.

\subsubsection*{Bounded Formats}

All the number formats described so far in this section
  represent an infinite set of dyadic numbers.
In practice, however,
  most number formats are designed to be encodable in
  a fixed-length bit string, which necessarily limits
  them to a finite number of representable values.
For example,
  the IEEE 754 standard explicitly defines
  floating-point formats that are encodable in
  32 bits (binary32) and 64 bits (binary64).
Because the set of representable values is finite,
  there must exist a largest representable magnitude.

Given a maximum representable value $b$,
  all representable values of a number format
  satisfy $|x| \leq b$.
However,
  some number formats are not symmetric:
  the largest positive and negative representable values
  may differ in magnitude.
For example,
  machine integer types such as C's \texttt{int}
  are asymmetric.
In such cases,
  separate bounds must be specified for
  positive and negative values.
If $b^{+}$ and $b^{-}$
  denote the positive and negative bounds,
  then $b^{-} \leq x \leq b^{+}$.

Rounding to a bounded format
  proceeds as in \Cref{ssec:bg:rounding}:
  we round without a maximum representable value, producing $y$;
  if $|y| > b$, then $\rnd_{F, rm}(x) = \omega(\mathrm{sgn}(y))$.

\begin{figure}[t!]
\begin{minted}[escapeinside=||]{Python}
def generate(|$p$|, |$q$|, |$b$|):
  assert |$p \in \Z_{\geq 1}$| or |$p == \infty$| # either floating point or fixed point
  assert is_power_of_two(|$q$|)
  |$v_{cut}$| = |$q \times 2^p$| # cutoff where the quantum increases
  |$v$| = 0
  yield |$v$| # start at 0
  while |$v + q \leq b$|: # always true if unbounded
    |$v$| = |$v + q$| # step by the quantum
    yield |$+v$|
    yield |$-v$|
    if is_power_of_two(|$v$|) and |$v \geq v_{cut}$|:
      |$q$| = |$2 \times q$|
\end{minted}
\caption{%
An algorithm for generating the representable values
  of fixed-point and floating-point number formats,
  with or without a maximum representable value.
}
\label{fig:format:generator}
\end{figure}

\subsection{Abstract Number Format}
\label{ssec:format:abstract}


Fundamentally,
  the generators in \Cref{ssec:format:generators} 
  differ in just one respect: whether the quantum increases.
Notice that if we set $p = \infty$
  for the floating-point generator,
  it becomes a fixed-point generator:
  the value at which the quantum increases,
  $v_{cut} = \infty$, is never reached.
This suggests that fixed-point and floating-point 
  formats are instances of a single parameterized format.
\Cref{fig:format:generator} shows
  the algorithm for this unified generator\footnote{%
  The implementation is surprisingly compact---%
  just 26 lines of Python using the Fraction class.}.


All number formats described by the generator
  are \emph{quantized}, in the sense that they
  have a minimum step size between representable values.
What number formats exist outside of
  this class of quantized number formats?
Many number formats in $\F$ can be described
  by generator parameters $(p, q, b)$,
  where $q = 0$ denotes no minimum quantum.
\Cref{fig:formats:taxonomy} classifies these formats
  by three properties of their parameters:
  \emph{quantized} ($q > 0$),
  \emph{precision-bound} ($p < \infty$),
  and \emph{bounded} ($b < \infty$);
  number formats may satisfy any combination of these properties.



\begin{figure}[t!]
\centering
\footnotesize
\begin{tikzpicture}
\def\r{1.1}        
\def\sep{1.2}      
\def\boxw{4}       
\def\boxh{4}       

\begin{scope}[shift={(0,0)}]
    \draw (-\boxw/2, -\boxh/2) rectangle (\boxw/2, \boxh/2);
    
    \coordinate (A) at (-\sep/2, 0.5);
    \coordinate (B) at (\sep/2, 0.5);
    \coordinate (C) at (0, -0.7);
    
    \draw (A) circle (\r);
    \draw (B) circle (\r);
    \draw (C) circle (\r);

    \begin{scope}[transparency group, opacity=0.4]
        \fill[blue] (A) circle (\r);
        \fill[blue] (C) circle (\r);
    \end{scope}
    
    \node at ($(A)+(-\r/2+0.1,0.3)$) {$p < \infty$};
    \node at ($(B)+(\r/2-0.1,0.3)$) {$b < \infty$};
    \node at ($(C)+(0,-\r/2)$) {$q > 0$};

    \node[overlay,anchor=west] at (2.6,1) {machine floating point (IEEE 754, MX)};
    \node[overlay,anchor=west] at (2.6,0) {machine fixed point (int8, int16)};
    \node[overlay,anchor=west] at (2.6,-1) {fixed point without overflow ($\Z$)};

    \node[overlay,anchor=east] at (-2.6,1) {MPFR floating point};
    \node[overlay,anchor=east] at (-2.6,0) {floating point without overflow};

    \draw[overlay,->] (2.5,1) -- (0,0); 
    \draw[overlay,->] (2.5,0) -- (0.6,-0.3); 
    \draw[overlay,->] (2.5,-1) -- (0,-0.75); 

    \draw[overlay,->] (-2.5,1) -- (-1,0.3); 
    \draw[overlay,->] (-2.5,0) -- (-0.6,-0.3); 
\end{scope}
\end{tikzpicture}
\caption{%
Classifying number formats by their generator parameters
  $(p, q, b)$ into \emph{quantized} ($q > 0$),
  \emph{precision-bound} ($p < \infty$),
  or \emph{bounded} ($b < \infty$) number formats.
Our abstract number format $\A$ captures
  number formats in the shaded region:
  either quantized ($q > 0$) or
  precision-bound ($p < \infty$) number formats.
}
\label{fig:formats:taxonomy}
\end{figure}

Of the common number formats in $\F$,
  few lack a minimum quantum.
Formal proofs of floating-point arithmetic
  often ignore the effects of overflow and underflow,
  allowing the exponent to be any integer---%
  these floating-point formats are neither quantized nor bounded\footnote{%
  MPFR~\cite{mpfr} would implement this class of number formats
    were its exponent not constrained by system limits.
  For most practical use cases,
    the exponent should be considered unbounded.
  }.
Posit numbers~\cite{posit} fall outside this classification
  since no choice of parameters describes them:
  their precision \emph{varies} with magnitude.
The remaining two regions---%
  unquantized, bounded floating-point formats
  and bounded dyadic numbers---%
  are absent in real-world usage;
  we are unsure of their potential use.



The three parameters $p$, $q$, and $b$
  determine the numerical values of an \emph{abstract number format}
  unifying fixed-point and floating-point formats.
Since every quantum is a power of two,
  we can equivalently represent it using its exponent, $exp$,
  where $q = 2^{exp}$; if the format has no minimum quantum,
  then $exp = -\infty$.
A fourth parameter, $S$,
  enumerates the \emph{special values} of the format.
We denote this abstract number format by $\A$
  and an instance by $\A(p, exp, b, S)$.
For any number format $\A(p, exp, b, S)$,
  we will call $p \in \Z_{\geq 1} \cup \{\infty\}$
  the \emph{maximum precision} of the format,
  $exp \in \Z \cup \{ -\infty \}$ the \emph{exponent of the minimum quantum},
  $b \in \A(p, exp, \infty) \cup \{ \infty \}$
  the \emph{maximum representable value},
  and $S \subseteq \{+\infty, -\infty, \mathrm{NaN}\}$
  the \emph{special values}\footnote{%
  We model a single NaN,
    ignoring payloads and the quiet/signaling distinction.}.
When a format has no special values,
  we write $\A(p, exp, b)$ for $\A(p, exp, b, \emptyset)$.
Notice that we could let $b$ be any dyadic number.
However,
  $\A(p, exp, b)$ is equivalent to any
  $\A(p, exp, b')$ such that $b \leq b' < \text{next}_{p, exp}(b)$
  where $\text{next}_{p, exp}(b)$ produces the subsequent
  representable value after $b$ in $\A(p, exp, \infty)$.
Thus,
  the restriction that $b \in \A(p, exp, \infty) \cup \{ \infty \}$
  does not affect which number formats we can represent
  and simplifies the theory.

We place three additional restrictions on the parameters of $\A$.
First, an instance must be quantized ($exp > -\infty$)
  or precision-bound ($p < \infty$),
  corresponding to the shaded region of \Cref{fig:formats:taxonomy}.
Second, an instance must be non-trivial,
  representing at least one nonzero value.
Third, we disallow $\A(1, -\infty, b)$,
  whose representable values are zero and $\pm 2^k$ for all $2^k \leq b$:
  every significand is $c = 1$ and there is no smallest positive value,
  so the even/odd classification of \Cref{ssec:bg:rounding}
  is ill-defined\footnote{%
  For $\A(p, -\infty, b)$ with $p \geq 2$,
    we classify every $\pm 2^k$ as even,
    ensuring that the remaining values alternate in parity.}.

\begin{table}[ht!]
\caption{%
Common number formats expressed as instances of $\A$.
The number format \texttt{fixed<$s$,$n$>}
  is a signed $n$-bit integer implicitly scaled by $2^s$.
Since $\A$'s bound is symmetric, signed integer and fixed-point types
  omit their most negative value (e.g., $-128$ for \texttt{int8\_t}).
}
\begin{tabular}{ll|lllll}
\toprule
$\A(p, exp, b, S)$ & Source & $p$ & $exp$ & $b$ & $b$ (approx.) & $S$ \\
\midrule
binary64 & IEEE 754~\cite{ieee754} & 53 & -1074 & $(2^{53} - 1) \times 2^{971}$ & $1.798 \times 10^{308}$ & $\{\pm\infty, \mathrm{NaN}\}$ \\
binary32 & IEEE 754 & 24 & -149 & $(2^{24} - 1) \times 2^{104}$ &
$3.403 \times 10^{38}$ & $\{\pm\infty, \mathrm{NaN}\}$ \\
TensorFloat-32~\cite{ampere} & Nvidia & 11 & -136 &
  $(2^{11} - 1) \times 2^{117}$ & $3.401 \times 10^{38}$ & $\{\pm\infty, \mathrm{NaN}\}$ \\
bfloat16~\cite{bfloat16} & Google Brain & 8 & -133 & $(2^8 - 1) \times 2^{120}$ & $3.390 \times 10^{38}$ & $\{\pm\infty, \mathrm{NaN}\}$ \\
E5M2 & OCP MX~\cite{ocp-mx} & 3 & -16 & $(2^3 - 1) \times 2^{13}$ & 57344 & $\{\pm\infty, \mathrm{NaN}\}$ \\
E4M3 & OCP MX & 4 & -9 & $(2^4 - 2) \times 2^{5}$ & 448 & $\{\mathrm{NaN}\}$ \\
E2M1 & OCP MX & 2 & -1 & $(2^2 - 1) \times 2^{1}$ & 6 & $\emptyset$ \\
\texttt{int8\_t} & - & $\infty$ & 0 & $(2^{7} - 1) \times 2^{0}$ & 127 & $\emptyset$ \\
\texttt{fixed<-4,8>} & - & $\infty$ & -4 & $(2^{7} - 1) \times 2^{-4}$ & 7.9375 & $\emptyset$ \\
\texttt{fixed<4,8>} & - & $\infty$ & 4 & $(2^{7} - 1) \times 2^{4}$ & 2032 & $\emptyset$ \\
$\Z$ & - & $\infty$ & 0 & $\infty$ & $\infty$ & $\emptyset$ \\
(default) & MPFR~\cite{mpfr} & 53 & $-\infty$ & $\infty$ & $\infty$ & $\{\pm\infty, \mathrm{NaN}\}$ \\
\bottomrule
\end{tabular}
\label{fig:formats}
\end{table}

\Cref{fig:formats} lists common number formats
  expressed as instances of our abstract number
  format $\A$.
If $F$ has a defining standard,
  the standard is listed alongside it.
The abstract number format $\A$
  covers a wide range of real number formats,
  including IEEE 754 floating-point formats,
  machine learning formats like
  Nvidia's TensorFloat-32 format~\cite{ampere},
  Google's bfloat16 format~\cite{bfloat16},
  and the OCP MX formats E5M2 and E4M3~\cite{ocp-mx}.
It also covers fixed-point formats,
  machine integer types,
  and arbitrary-precision formats like
  the integers, $\Z$, and MPFR floating-point numbers.
Formats encode different special values:
  the IEEE 754 formats and their derivatives have
  both infinities and NaN, as does MX E5M2;
  MX E4M3 has only NaN;
  and MX E2M1 and the integer and fixed-point formats have neither.
Importantly,
  the parameters exactly determine the set
  of representable values within the number format.

\subsubsection*{Comparison to Flocq}

Flocq~\cite{flocq} takes a more general approach
  to unifying fixed- and floating-point formats.
Its \emph{generic format} is parameterized by a radix $\beta$
  and an \emph{exponent function} $\mathrm{fexp}$
  that sets the quantum between consecutive powers of $\beta$:
  a real number $x$ with $|x| \in [\beta^{e-1}, \beta^{e})$
  is representable when $x$ is an integer multiple of $\beta^{\mathrm{fexp}(e)}$.
Fixed point (FIX) is $\mathrm{fexp}(e) = exp$;
  floating point (FLX) with unbounded exponents is
  $\mathrm{fexp}(e) = e - p$;
  and floating point with a minimum quantum (FLT), as in IEEE 754,
  is $\mathrm{fexp}(e) = \max(e - p, exp)$.
In radix two,
  $\A(p, exp, \infty)$ is equivalent to the generic format with
  $\mathrm{fexp}(e) = \max(e - p, exp)$,
  which is FIX when $p = \infty$ and FLX when $exp = -\infty$;
  $\A(p, exp, b)$ further restricts it to $|x| \leq b$.
In \Cref{sec:round},
  we show how our particular parameterization
  reduces format containment, and with it correct double rounding,
  to an efficient check on $p$, $exp$, and $b$.

\subsubsection*{Format Redundancy}

The parameters \emph{do not uniquely} determine
  the set of representable values;
  there exists a redundancy in specifying number formats.
Consider the fixed-point format $\A(\infty, -2, 1)$
  and the floating-point format $\A(2, -2, 1)$,
  each a subset of the fixed-point and floating-point
  number format examples in \Cref{ssec:format:generators},
  respectively.
If we enumerate the representable values,
  we obtain the set $\{ 0, \pm 1/4, \pm 1/2, \pm 3/4, \pm 1 \}$
  for \emph{both} number formats.
We see that the quantum never increases
  for the floating-point format since the cutoff $v_{cut} = 1$
  is never exceeded.



We therefore classify $\A(2, -2, 1)$
  as a \emph{degenerate} floating-point format
  since the precision constraint $p = 2$ is unused.
In fact,
  $\A(2, -2, 1) = \A(3, -2, 1) = \cdots = \A(\infty, -2, 1)$.
Thus,
  we see another link between fixed-point and floating-point formats:
  \emph{every bounded, fixed-point number format
  is equivalent to infinitely many degenerate,
  floating-point number formats}.
Consequently, the classes of \Cref{fig:formats:taxonomy}
  follow the chosen parameters rather than the set of values:
  $\A(\infty, -2, 1)$ is not precision-bound, yet $\A(2, -2, 1)$ is.

This observation is particularly relevant
  for emerging machine learning number formats,
  where there is increasing pressure to reduce bitwidths for efficiency.
For example,
  the MX E2M1 format,
  $\A(2, -1, 6)$, is only four bits wide.
As number formats continue to shrink,
  implementers must contend with floating-point formats
  essentially becoming fixed-point formats.
At the extreme,
  models such as BitNet~\cite{bitnet, bitnet2} achieve
  competitive accuracy with FP16 transformer architectures
  while constraining weights to $\{-1, +1\}$ or $\{-1, 0, +1\}$.

\section{Format Containment and Correct Double Rounding}
\label{sec:round}

Having defined an abstract number format, $\A$,
  we now present two applications.
First, for any two instances of $\A$,
  we give an efficient, complete check
  for deciding \emph{format containment}---%
  that is, whether one number format is a
  subset of the other (\Cref{ssec:round:contain}).
Second,
  we use format containment to characterize
  when double rounding is correct
  (\Cref{ssec:round:safe}).
Finally,
  we show how correct double rounding enables a
  modular approach to building correctly-rounded number libraries
  by \emph{extending} existing number libraries
  to support new number formats (\Cref{ssec:round:design}).
We have mechanized this section's results in Lean~\cite{lean4};
  Appendix A of our supplementary materials presents prose proofs.

\subsection{Format Containment}
\label{ssec:round:contain}

How can we determine whether one number format is a subset 
  of another?
Naively, since format containment is just the subset
  relation $F_1 \subseteq F_2$,
  we could check if every value
  in $F_1$ is in $F_2$.
In particular,
  if both $F_1$ and $F_2$ are quantized instances
  of $\A$, then we could run the generators
  described in \Cref{ssec:format:generators}, enumerate 
  all possible values in both formats,
  and check if each value in $F_1$ is also in $F_2$.
For small, finite number formats---like those in 
  machine learning---this approach wouldn't 
  necessarily take a long time\footnote{%
  For formats with infinitely many representable values,
  enumeration is impossible.}.
Instead, we present an efficient check
  directly on the parameters of $\A$,
  without enumerating representable values.



Let $F_1 = \A(p_1, exp_1, b_1, S_1)$ and $F_2 = \A(p_2, exp_2, b_2, S_2)$.
\Cref{fig:round:containment} presents
  the two judgments for determining when $F_1 \subseteq F_2$.
The first rule, \textsc{$\A$-Contains-Prec}
  covers the majority of the cases,
  but is too restrictive for degenerate floating-point formats
  (or equivalently, bounded fixed-point formats).
For these cases, \textsc{$\A$-Contains-Sub}
  provides an alternative containment test.

\begin{figure}
\begin{mathpar}
%
\inferrule*[right=$\A$-Contains-Prec]
  { p_1 \leq p_2 \\ exp_1 \geq exp_2 \\ b_1 \leq b_2 \\ S_1 \subseteq S_2 }
  { \A(p_1, exp_1, b_1, S_1) \subseteq \A(p_2, exp_2, b_2, S_2) }

\inferrule*[right=$\A$-Contains-Sub]
  { b_1 \leq 2^{exp_1 + p_2} \\ exp_1 \geq exp_2 \\ b_1 \leq b_2 \\ S_1 \subseteq S_2 \\ exp_1 > -\infty }
  { \A(p_1, exp_1, b_1, S_1) \subseteq \A(p_2, exp_2, b_2, S_2) }
\end{mathpar}
\caption{%
Format containment rules
}
\label{fig:round:containment}
\end{figure}

\textsc{$\A$-Contains-Prec} says: if $F_2$ has at least 
  as much precision ($p_1 \leq p_2$), at least as small a 
  minimum quantum ($exp_1 \geq exp_2$), at least as 
  large a bound ($b_1 \leq b_2$),
  and every special value of $F_1$ ($S_1 \subseteq S_2$),
  then $F_1 \subseteq F_2$.
The conditions $b_1 \leq b_2$ and $S_1 \subseteq S_2$ are clearly necessary;
  otherwise, $b_1 \not\in F_2$ or $S_1 \not\subseteq F_2$.
Since every value in $F_2$ is a multiple of $2^{exp_2}$,
  it must be the case that $exp_1 \geq exp_2$.
Finally,
  $p_1 \leq p_2$ captures the relative precision of the 
  two formats.
Intuitively,
  any value with a $p_1$-digit binary significand is
  also expressible with a $p_2$-digit significand
  by padding with trailing zero digits and
  adjusting the exponent accordingly.
Therefore,
  under these four conditions,
  every value representable in $F_1$ is also representable 
  in $F_2$.

For degenerate floating-point formats,
  the maximum precision has no effect on the set of representable 
  values,
  yet the precision test of \textsc{$\A$-Contains-Prec} may fail.
In these cases,
  \textsc{$\A$-Contains-Sub} applies a different test:
  $b_1 \leq 2^{exp_1 + p_2}$.
Every representable value of $F_1$ is a multiple of $2^{exp_1}$,
  and $F_2$ represents every multiple of $2^{exp_1}$
  up to $2^{exp_1 + p_2}$.
Therefore,
  $b_1 \leq 2^{exp_1 + p_2}$ ensures that every representable
  value of $F_1$ is also representable in $F_2$.

The two rules are also \emph{complete}:

\begin{lemma}
\label{lem:contain-complete}
For any instances $F_1$ and $F_2$ of $\A$,
  $F_1 \subseteq F_2$ if and only if
  \textsc{$\A$-Contains-Prec} or \textsc{$\A$-Contains-Sub} applies.
\end{lemma}

We now prove that if $F_1 \subseteq F_2$,
  then one of the two rules applies.
Since $b_1 \in F_1$ and $S_1 \subseteq F_1$,
  we have $b_1 \leq b_2$ and $S_1 \subseteq S_2$;
  since $F_1$ is non-trivial, $2^{exp_1} \in F_1$,
  so $exp_1 \geq exp_2$.
If $p_1 \leq p_2$, \textsc{$\A$-Contains-Prec} applies.
Otherwise, the value $(2^{p_2} + 1) \times 2^{exp_1}$
  needs $p_2 + 1$ digits, so it lies outside $F_2$
  and must therefore lie outside $F_1$'s bound;
  as $b_1$ is representable in $\A(p_1, exp_1, \infty)$,
  this forces $b_1 \leq 2^{exp_1 + p_2}$,
  and \textsc{$\A$-Contains-Sub} applies.
Appendix A gives the full proof,
  including the cases $b_1 = \infty$ and $exp_1 = -\infty$.

Through the correspondence with Flocq's generic format
  (\Cref{ssec:format:abstract}),
  the inclusion lemmas\footnote{%
  \texttt{generic\_inclusion\_mag} and \texttt{generic\_inclusion\_le}, Flocq~4.2.2.}
  of Flocq~\cite{flocq} also establish the soundness of both rules,
  but not their completeness,
  and leave the user a side condition on the formats' exponent functions to prove;
  \Cref{fig:round:containment} reduces it to comparisons of parameters,
  including the bound.
Moreover,
  both of these tests run in constant-time relative
  to the number of representable values in $\A(p_1, exp_1, b_1)$,
  making them amenable to static analysis.
The parameters of $\A$
  bound the set of representable values and the distribution
  of representable values.

\subsection{Correct Double Rounding}
\label{ssec:round:safe}

Recall from \Cref{ssec:bg:double-round} that, 
  for number formats $F_1, F_2$ and rounding modes $rm_1, rm_2$,
  the double rounding $\rnd_{F_1, rm_1}\circ \rnd_{F_2, rm_2}$
  is correct whenever $\rnd_{F_1, rm_1} = \rnd_{F_1, rm_1} \circ \rnd_{F_2, rm_2}$:
  when $F_2$ has no observable effect on the final answer. 
When is this true?
Intuitively,
  it should be true when the intermediate rounding
  does not perturb its input too much,
  so that the final rounding makes the same choice
  it would have made on the original value.
In this section,
  we show how this intuition can be formalized as
  a format containment check, with the required precision
  depending on the rounding modes involved.


We begin by characterizing how much information a rounding preserves; this
  directly determines the preconditions for correct double rounding.
Rounding a real number $x$ to an instance of $\A$
  proceeds in two steps:
  split the significand of $x$ at a position determined by
  $x$, $p$, and $exp$---%
  we say the rounding \emph{rounds to} the $w$ digits it retains---%
  and then adjust the retained digits based on
  the discarded digits and the rounding mode
  (e.g., forcing the last digit odd under RTO).
The rounding mode plays no role in the first step.
Formally:
\begin{lemma}
\label{lem:kept-digits}
Let $\A(p, exp, b)$ be a number format,
  $rm$ be any rounding mode,
  and $x \in \R$ be a real number.
The number of digits $\rnd_{\A(p, exp, b), rm}(x)$ rounds to
  is a function of $p, exp$, and $x$.
\end{lemma}

This allows us to write $w_1$ and $w_2$ for the number of digits
  we round to when rounding to $F_1$ and $F_2$. 
\Cref{lem:kept-cmp} provides the relationship
  between $w_1$ and $w_2$:

\begin{lemma}
\label{lem:kept-cmp}
Let $x \in \R$ be a nonzero real number and $k \geq 0$ be an integer.
Suppose rounding to $\A(p, exp, b)$ rounds to $w_1$ digits.
Rounding to $\A(p + k, exp - k, b)$ rounds to $w_2 = w_1 + k$ digits.
\end{lemma}

From here, it follows naturally that, 
  if $F_1 \subseteq F_2$, then $w_1 \leq w_2$.
For double roundings where both rounding modes are the same---%
  specifically for RAZ, RTO, and RTZ---%
  this is sufficient; correct double rounding reduces to a
  simple format containment check on $F_1$ and $F_2$.
However, in cases where RTO is used as the intermediate rounding and
  a different rounding mode---RTZ, RAZ, or RNE---is the final
  rounding mode, the picture is slightly more complicated. 
In these cases,
  we add one or two bits of padding to the intermediate format, 
  in accordance with the following lemma:

\begin{lemma}
\label{lem:rto-rounding}
Let $x'$ be the result of rounding $x$
  to precision $w + k$ under RTO,
  where $k \geq 1$.
Then $x$ is representable at precision $w$
  if and only if $x'$ is representable at precision $w$.
If $x$ is unrepresentable at precision $w$
  and $y_1 < x < y_2$ are the two adjacent representable values
  at precision $w$ that are closest to $x$,
  then $y_1 < x' < y_2$.
\end{lemma}


There are sixteen pairings of rounding modes for
  the intermediate and final roundings
  over RNE, RAZ, RTZ, and RTO.
Of these,
  \Cref{fig:correct-double-round} lists the six combinations
  that guarantee correct double rounding,
  along with the conditions under which they hold.
The first six rules decide when rounding to $F_2$ under $rm_2$
  is a \emph{safe intermediate} for rounding to $F_1$ under $rm_1$,
  written $\langle F_1, rm_1 \rangle \rhd \langle F_2, rm_2 \rangle$:
  rounding a real number through $F_2$ matches rounding it directly
  to $F_1$, or both overflow.
Their bounds on $b_2$ ensure that
  $F_2$ overflows only when $F_1$ does.

\textsc{rnd-Double} extends correctness to every input
  when special values compose ($\rnd_{F_1, rm_1} \circ \sigma_2 = \sigma_1$)
  and overflow composes ($\rnd_{F_1, rm_1} \circ \omega_2 = \omega_1$).
IEEE 754's $\sigma$ and $\omega$ compose for every rule but \textsc{rnd-RTO-RTZ}:
  the intermediate RTO rounding overflows to $\pm\infty$
  (our convention, since IEEE 754 has no RTO)
  and the final RTZ rounding returns $\pm\infty$ unchanged,
  whereas rounding directly under RTZ overflows to $\pm b_1$;
  an intermediate that saturates instead ($\omega_2(\pm 1) = \pm b_2$) restores it.
Longer chains of roundings reduce pairwise:
  if rounding through $F_3$ into $F_2$
  and through $F_2$ into $F_1$ each satisfy \textsc{rnd-Double},
  then $\rnd_{F_1, rm_1} \circ \rnd_{F_2, rm_2} \circ \rnd_{F_3, rm_3} = \rnd_{F_1, rm_1}$.

\begin{figure}[!t]
\begin{mathpar}
\inferrule*[right=rnd-RTZ-RTZ]
  { \A{(p_1, exp_1, \text{next}_{p_1, exp_1}(b_1), S_1)} \subseteq \A{(p_2, exp_2, b_2, S_2)} }
  { \langle \A(p_1, exp_1, b_1, S_1), \text{RTZ} \rangle \rhd \langle \A(p_2, exp_2, b_2, S_2), \text{RTZ} \rangle }

\inferrule*[right=rnd-RAZ-RAZ]
  { \A{(p_1, exp_1, b_1, S_1)} \subseteq \A{(p_2, exp_2, b_2, S_2)} }
  { \langle \A(p_1, exp_1, b_1, S_1), \text{RAZ} \rangle \rhd \langle \A(p_2, exp_2, b_2, S_2), \text{RAZ} \rangle }

\inferrule*[right=rnd-RTO-RTO]
  { \A{(p_1, exp_1, \text{next}_{p_1, exp_1}(b_1), S_1)} \subseteq \A{(p_2, exp_2, b_2, S_2)} \\ p_2 \geq 2 }
  { \langle \A(p_1, exp_1, b_1, S_1), \text{RTO} \rangle \rhd \langle \A(p_2, exp_2, b_2, S_2), \text{RTO} \rangle }

\inferrule*[right=rnd-RTO-RTZ]
  { \A{(p_1 + 1, exp_1 - 1, \text{next}_{p_1, exp_1}(b_1), S_1)} \subseteq \A{(p_2, exp_2, b_2, S_2)} }
  { \langle \A(p_1, exp_1, b_1, S_1), \text{RTZ} \rangle \rhd \langle \A(p_2, exp_2, b_2, S_2), \text{RTO} \rangle }

\inferrule*[right=rnd-RTO-RAZ]
  { \A{(p_1 + 1, exp_1 - 1, \text{next}_{p_1, exp_1}(b_1), S_1)} \subseteq \A{(p_2, exp_2, b_2, S_2)} }
  { \langle \A(p_1, exp_1, b_1, S_1), \text{RAZ} \rangle \rhd \langle \A(p_2, exp_2, b_2, S_2), \text{RTO} \rangle }

\inferrule*[right=rnd-RTO-RNE]
  { \A{(p_1 + 2, exp_1 - 2, \text{next}_{p_1 + 1, exp_1 - 1}(b_1), S_1)} \subseteq \A{(p_2, exp_2, b_2, S_2)} }
  { \langle \A(p_1, exp_1, b_1, S_1), \text{RNE} \rangle \rhd \langle \A(p_2, exp_2, b_2, S_2), \text{RTO} \rangle }

\inferrule*[right=rnd-Double]
  { \langle F_1, rm_1 \rangle \rhd \langle F_2, rm_2 \rangle \\
    \rnd^{\sigma_1, \omega_1}_{F_1, rm_1} \circ \sigma_2 = \sigma_1 \\
    \rnd^{\sigma_1, \omega_1}_{F_1, rm_1} \circ \omega_2 = \omega_1 }
  { \rnd^{\sigma_1, \omega_1}_{F_1, rm_1} =
    \rnd^{\sigma_1, \omega_1}_{F_1, rm_1} \circ \rnd^{\sigma_2, \omega_2}_{F_2, rm_2} }
\end{mathpar}
\caption{%
Correct double rounding.
The first six rules check that rounding a real number through $F_2$
  matches rounding it directly to $F_1$, or both overflow
  ($\langle F_1, rm_1 \rangle \rhd \langle F_2, rm_2 \rangle$).
Rule \textsc{rnd-Double} adds that special values and overflow compose.
}
\label{fig:correct-double-round}
\end{figure}

The first three rules
  (\textsc{rnd-RTZ-RTZ}, \textsc{rnd-RAZ-RAZ}, \textsc{rnd-RTO-RTO})
  apply when both roundings use the same mode
  and reduce to a single containment check.
\textsc{rnd-RTZ-RTZ} and \textsc{rnd-RAZ-RAZ} are well known;
  CPFloat~\cite{cpfloat}, for example, notes that double rounding
  under a directed rounding mode is always correct.
For RTZ, truncating to $w_2$ digits then $w_1$ digits
  equals truncating directly to $w_1$ when $w_1 \leq w_2$---%
  guaranteed by \Cref{lem:kept-cmp}.
Similar arguments apply to the RAZ and RTO modes.
The \textsc{rnd-RTO-RTO} rule also requires $p_2 \geq 2$
  to avoid parity mismatches when
  \emph{both} formats have one bit of precision.

The remaining rules---%
  \textsc{rnd-RTO-RTZ}, \textsc{rnd-RTO-RAZ}, and \textsc{rnd-RTO-RNE}---%
  all use RTO as the intermediate rounding.
\citet{boldo-rto} prove \textsc{rnd-RTO-RNE}
  and remark, without proof, that the same intermediate serves
  any final rounding mode;
  RLIBM-ALL~\cite{rlibm-all} applies this with two extra bits
  for every mode.
Our rules make this precise and show that directed final modes
  need only one extra bit, and RTO none.
The first two rules
  share the same premise, ensuring $w_1 + 1 \leq w_2$.
Under that condition,
  \Cref{lem:rto-rounding} guarantees that any
  (un)representable value in $\A{(p_1, exp_1, b_1)}$ is still
  (un)representable after the intermediate rounding.

For \textsc{rnd-RTO-RNE},
  a slightly different premise instead
  guarantees that $w_1 + 2 \leq w_2$.
The additional digit preserves halfway cases,
  which is necessary to correctly round under RNE.
Flocq mechanizes their theorem as \texttt{round\_N\_odd}~\cite{flocq}.
Stated over $\A$ with our bound condition, it reads:
  for floating-point formats
  $\A{(p_1, exp_1, b_1)}$ and $\A{(p_2, exp_2, b_2)}$,
  if $2 \leq p_1 \leq p_2 - 2$, $exp_1 \geq exp_2 + 2$, and
  $\text{next}_{p_1+1, exp_1-1}(b_1) \leq b_2$, then
\[
\rnd_{\A(p_1, exp_1, b_1), \text{RNE}} =
  \rnd_{\A(p_1, exp_1, b_1), \text{RNE}} \circ
  \rnd_{\A(p_2, exp_2, b_2), \text{RTO}}.
\]
Their bound condition is slightly weaker:
  $b_2$'s overflow threshold need only be at least $b_1$'s.

The other ten pairings of rounding modes
  do not support correct double rounding in general.
The construction applies to any $F_1 \subseteq F_2$ assuming
  that no rounding overflows.
Let $y$ be an even representable value of $F_1$
  with odd neighbors $y^- < y < y^+$,
  and let $\delta$ be a quarter of the quantum between
  $F_2$'s representable values on the relevant side of $y$.
For each pairing (\Cref{fig:double-round-cex}),
  $x$ sits just past $y$, or just past the halfway point to a neighbor.
The intermediate rounding maps $x$ to $y$, or to the $y$-side of the
  halfway point, so the double rounding returns the even $y$ rather than
  the odd neighbor that direct rounding gives.
Under the \textsc{RTZ-RAZ} pairing, for instance, $x = y + \delta$ truncates
  to $y$ under RTZ and stays there under RAZ, whereas
  $\rnd_{F_1, \text{RAZ}}(x) = y^+$.
Since $\delta$ shrinks as $F_2$ grows, $x$ only approaches $y$:
  a larger intermediate format never helps.

\begin{table}[!t]
\centering
\small
\caption{%
Counterexamples for the ten unsound double roundings $rm_2$-$rm_1$
  ($rm_2$ intermediate, $rm_1$ final).
Here $y$ is an even value of $F_1$ with odd neighbors $y^- < y < y^+$,
  $\delta$ is a quarter of $F_2$'s spacing near $y$, and $m$ is the halfway
  point just above $y$ (for \textsc{RTZ-RNE}) or just below it
  (for \textsc{RAZ-RNE}, \textsc{RNE-RNE}).
All require only $F_1 \subseteq F_2$,
  except \textsc{RNE-RNE}, which requires
  $\A(p_1 + 1, exp_1 - 1, b_1) \subseteq \A(p_2, exp_2, b_2)$
  to ensure $m \in F_2$.
}
\begin{tabular}{@{}llll@{}}
  \toprule
  $x$ & $\rnd_{F_1, rm_1}(x)$ & $\rnd_{F_1, rm_1}(\rnd_{F_2, rm_2}(x))$ & Double roundings $rm_2$-$rm_1$ \\
  \midrule
  $y - \delta$ & $y^-$ & $y$
    & \textsc{RNE-RTZ}, \textsc{RAZ-RTZ}, \textsc{RAZ-RTO}, \textsc{RNE-RTO} \\
  $y + \delta$ & $y^+$ & $y$
    & \textsc{RNE-RAZ}, \textsc{RTZ-RAZ}, \textsc{RTZ-RTO} \\
  $m \pm \delta$ & $y^{\pm}$ & $y$
    & \textsc{RTZ-RNE}, \textsc{RAZ-RNE}, \textsc{RNE-RNE} \\
  \bottomrule
\end{tabular}
\label{fig:double-round-cex}
\end{table}

As the intermediate rounding,
  RTO supports correct double rounding for all four final modes,
  at the cost of at most two extra bits of precision.
The directed modes RTZ and RAZ support double rounding
  only when the intermediate and final modes agree, and
  RNE---despite being the default, and often the only,
  rounding mode available in programming environments---%
  permits it only in trivial cases
  (e.g., $F_2 \cap [-b_1, b_1] = F_1$, on inputs where neither rounding overflows),
  even when the intermediate format is made arbitrarily large.

Prior work exploits this property,
  using RTO whenever intermediate results must later
  be safely re-rounded~\cite{rlibm,gcc-rto,boldo-rto-apps}.
We appeal to this same property to build a modular number library:
  a single round-to-odd kernel, evaluated in high enough precision,
  re-rounds to a correctly-rounded result under any final rounding mode.
Rather than a separate implementation for each,
  \name computes its high-precision intermediates in RTO (\Cref{sec:eval}).

\subsection{Modular Number Libraries}
\label{ssec:round:design}

Whenever double rounding is correct,
  we can implement a correctly-rounded function
  by combining higher-precision computation with
  rounding to the target number format:
\[
\rnd_{F_1, rm_1} = \rnd_{F_1, rm_1} \circ \rnd_{F_2, rm_2}
\Rightarrow
f_{F_1, rm_1} = \rnd_{F_1, rm_1} \circ f_{F_2, rm_2}.
\]
Thus correct double rounding enables
  a modular approach to building number libraries
  by cleanly \emph{separating rounding from arithmetic}.
For each real-valued function $f$,
  it suffices to provide a high-precision implementation
  $f_{F_2, rm_2}$ that can be safely re-rounded
  to the target number formats.
Importantly,
  the rules from \Cref{fig:correct-double-round}
  apply \emph{independently} of the particular function $f$
  unlike in prior work~\cite{rump-double-round,roux-double-round}.
For the intermediate rounding,
  the RTO rounding mode is the natural choice.

In particular,
  the high-precision computation can
  rely on \emph{existing} implementations
  such as hardware-provided capabilities,
  well-tested number libraries like \softfl and MPFR,
  or custom kernels.
As long as these implementations
  produce a correctly-rounded result that satisfies
  the preconditions for correct double rounding,
  they can be safely composed with rounding
  to lower-precision formats.
Compared to a traditional monolithic design,
  this method yields a number library that
  is easier to maintain, requiring less code;
  easier to test, since only the rounding operators must be validated;
  and efficient, by reusing highly optimized high-precision
  implementations.

In \Cref{sec:eval},
  we present \name, a prototype of this approach.
We describe how it extends
  several high-precision implementations
  and evaluate its performance against existing number libraries.

\section{Program Analyses and Transformations}
\label{sec:analysis}




The results of the previous sections---format containment
  and correct double rounding---operate on individual
  rounding operations.
We now turn to programs that compose \textit{multiple}
  correctly-rounded operations.
We introduce \emph{format inference},
  which conservatively bounds each expression by an instance of $\A$
  (\Cref{ssec:analysis:infer}),
  and use it to justify transformations that
  eliminate and insert roundings
  (\Cref{ssec:analysis:round_elim}).
Together,
  these yield a \emph{semantics-preserving} rewriting strategy
  for mapping numerical specifications onto a target
  environment's rounded operations (\Cref{ssec:analysis:compile});
  we apply it by hand to a case study in \Cref{sec:case_study}.


%
%

\subsection{Format Inference}
\label{ssec:analysis:infer}


We consider programs consisting of assignment statements whose 
  right-hand sides are constants, variables, or 
  applications of arithmetic operators (e.g., addition). 
Each operator is either correctly-rounded, $f_{F,rm}$,
  or unrounded, $f_{\R,\_}$, that is, it computes over the reals and does not round.
For example, $\text{add}_{\text{FP32},\text{RNE}}$
  denotes correctly-rounded single-precision addition, 
  while $\text{mul}_{\R,\_}$ denotes exact multiplication.


Format inference proceeds by analyzing
  each assignment statement in execution order, maintaining a
  mapping from variables to number formats.
Numerical constants are assigned the
  tightest bound containing them; variables take the format
  they are mapped to. 
Correctly-rounded operations $f_{\A(p, exp, b, S), rm}$ 
  produce values in $\A(p, exp, b, S)$, and unrounded operations 
  $f_{\R, \_}$ are assigned the ``format" $\R$.

For unrounded operations,
  we can often derive a tighter bound from
  the formats of the arguments.
Negation and absolute value preserve the format of their operand,
  except that they negate or take the magnitude of its infinities.
Addition and multiplication are as follows:
\[
\text{fmt}(\text{add}_{\R, \_}(e_1, e_2)) =
  \text{fmt}(e_1) \oplus \text{fmt}(e_2),
\qquad
\text{fmt}(\text{mul}_{\R, \_}(e_1, e_2)) =
  \text{fmt}(e_1) \otimes \text{fmt}(e_2).
\]
Here, $\oplus$ and $\otimes$ are given by:
\begin{align*}
\A_{(p_1, exp_1, b_1, S_1)} \oplus \A_{(p_2, exp_2, b_2, S_2)} &\subseteq
    \A_{(\max(1,\, \lceil \log_2(\lfloor (b_1 + b_2)/2^{\text{min}(exp_1, exp_2)} \rfloor + 1) \rceil),\ \text{min}(exp_1, exp_2),\ b_1 + b_2,\ S_\oplus)}, \\
\A_{(p_1, exp_1, b_1, S_1)} \otimes \A_{(p_2, exp_2, b_2, S_2)} &\subseteq
  \A_{(p_1 + p_2, exp_1 + exp_2, b_1 \times b_2, S_\otimes)}.
\end{align*}
For addition, the precision is the number of binary digits
  needed to represent the worst-case sum, $b_1 + b_2$, as
  a multiple of $2^\text{min}(exp_1, exp_2)$.
The sets $S_\oplus$ and $S_\otimes$ hold the special values
  that a sum or product can produce under IEEE 754:
  a NaN operand yields NaN,
  $\infty - \infty$ and $\infty \times 0$ are NaN,
  and an infinity otherwise yields an infinity.
For example, $\otimes$ takes two values in the MX E5M2 format
  $\A(3, -16, 7 \times 2^{13}, \{\pm\infty, \mathrm{NaN}\})$
  to $\A(6, -32, 49 \times 2^{26}, \{\pm\infty, \mathrm{NaN}\})$;
  \Cref{sec:case_study} traces this analysis through an entire program.

This analysis is similar to \emph{interval analysis}~\cite{interval},
  which tracks intervals $[l, h]$ for each real-valued expression.
Format inference is more precise, tracking not only the range
  of possible values but also their spacing, since only some values
  in the interval are representable.

\subsection{Introducing \& Eliminating Roundings}
\label{ssec:analysis:round_elim}

Format inference yields a single insight that enables
  us to \emph{insert} and \emph{eliminate} roundings.
Since every correctly-rounded operation consists of a real-valued
  computation followed by a rounding 
  $f_{F, rm} = \rnd_{F, rm} \circ f_{\R, \_}$,
  we can apply format inference to the unrounded computation
  $f_{\R, \_}$ \textit{prior} to rounding.
Suppose format inference determines
  that the unrounded result already lies
  within the target format $F$; formally,
\[
\text{fmt}(f_{\R, \_}(e_1, \ldots, e_n)) \subseteq F.
\]
In this case,
  every value $v$ produced by $f_{\R, \_}$ is already representable in $F$,
  so $\rnd_{F, rm}(v) = v$, assuming $\sigma = \text{id}$, and 
\[
  f_{F, rm}(e_1, \ldots, e_n) = f_{\R, \_}(e_1, \ldots, e_n).
\]
This identity can be read in two directions. 

Left-to-right, it justifies \textit{canonicalization}:
  a correctly-rounded operation can be replaced by its real
  counterpart, since the rounding does nothing,
  simplifying the program.
For example,
  consider the expression $\text{mul}_{\text{FP32}, \text{RNE}}(e_1, e_2)$,
  where $e_1$ and $e_2$ produce values representable in MX E5M2.
In the previous subsection,
  format inference discovers that $\text{mul}_{\R, \_}(e_1, e_2)$
  is bounded by $\A(6, -32, 49 \times 2^{26}, \{\pm\infty, \mathrm{NaN}\}) \subseteq \text{FP32}$.
Therefore,
  the rounding to $\text{FP32}$ has no effect,
  and so the following rewrite is semantics-preserving:
\[
  \text{mul}_{\text{FP32}, \text{RNE}}(e_1, e_2)
  \;\Longrightarrow\;
  \text{mul}_{\R, \_}(e_1, e_2).
\]

Read right-to-left, it justifies \textit{finitization}:
  an unrounded operation can be replaced by a correctly-rounded
  one, since the rounding would change nothing semantically.
This transformation enables a choice of number formats
  and rounding modes which are available in the programming environment:
  real-valued computations can be replaced by equivalent roundings.
Running our previous example in reverse,
  if we have $\text{mul}_{\R, \_}(e_1, e_2)$,
  where $e_1$ and $e_2$ produce values representable in MX E5M2,
  then we can rewrite the unrounded product into an
  FP32 operation available in hardware,
\[
\text{mul}_{\R, \_}(e_1, e_2)
  \;\Longrightarrow\;
  \text{mul}_{\text{FP32}, \text{RNE}}(e_1, e_2)
\]
  with no change in result.


The \emph{correct double rounding} rules
  presented in \Cref{ssec:round:safe}
  lead to a similar bi-directional rewrite.
In the canonicalization direction,
  we can detect correct double roundings within the program
  and merge them into a single correctly-rounded operation.
Consider a composition
  $\rnd_{F_1, rm_1}(f_{F_2, rm_2}(e))$.
If we prove that it is a correct double rounding,
  i.e., $\rnd_{F_1, rm_1} \circ \rnd_{F_2, rm_2} = \rnd_{F_1, rm_1}$,
  then the intermediate rounding to $F_2$
  does not affect the result of the expression.
We can therefore replace the composition
  with the single correctly-rounded operation $f_{F_1, rm_1}(e)$.
In the finitization direction,
  we can expand a correctly-rounded operation
  into a higher-precision computation and re-rounding,
  both of which are available in the programming environment.



\subsection{A Rewriting Strategy}
\label{ssec:analysis:compile}




\begin{wrapfigure}{r}{0.4\textwidth}
\centering
\vspace{-\baselineskip}
\footnotesize
\begin{tikzpicture}[>=stealth, node distance=2cm]
  \node (N1) at (0,0) {$P_1$};
  \node (N2) at (1.5,1.5) {$P_2$};
  \node (N3) at (3,0) {$P_3$};
  \draw[->] (N1) -- (N2);
  \draw[->] (N2) -- (N3);
  \draw[->, double, dotted, double distance=1pt] (N1) -- (N3);
  \node at (0,0.9) {canonicalize};
  \node at (2.7,0.9) {finitize};
\end{tikzpicture}
\caption{%
Rewriting using \emph{canonicalization} and \emph{finitization}.
All three programs, $P_1$, $P_2$, and $P_3$, are equivalent.
}
\label{fig:analysis:compile}
\vspace{-\baselineskip}
\end{wrapfigure}

Together, format inference and the transformations of
  \Cref{ssec:analysis:round_elim} yield a rewriting strategy
  for programs of correctly-rounded operations.
The strategy has three phases (\Cref{fig:analysis:compile}).
First, format inference assigns a format to every expression.
Second, canonicalization lifts the program $P_1$ to
  a program $P_2$ in which every rounding is \emph{essential},
  replacing $f_{F, rm}(e)$ with $f_{\R, \_}(e)$
  whenever $\text{fmt}(f_{\R, \_}(e)) \subseteq F$.
Third, finitization lowers $P_2$ to an equivalent program $P_3$
  in which every rounding is \emph{available}
  in the programming environment,
  replacing each unrounded operation with a rounded one
  whose available format $F'$ contains its inferred format.
Since the inferred format need not itself be available,
  $F'$ may be larger; we call this \emph{promotion}.
Unlike $P_1$ and $P_2$,
  $P_3$ uses only available operations
  and is thus fully executable.

Format inference also guides \emph{exact decompositions},
  which rewrite an expression as a sum of parts,
  in the spirit of error-free transformations~\cite{eft}.
For example, \emph{splitting} performs the rewrite
\[
  e \;\Longrightarrow\; \text{add}_{\R, \_}(\text{hi}_k(e), \text{lo}_k(e)),
\]
  where $\text{hi}_k(e)$ truncates $e$ to a multiple of $2^{exp + k}$
  and $\text{lo}_k(e)$ is the remainder;
  if $\text{fmt}(e) = \A(p, exp, b)$ with $exp > -\infty$,
  then $\text{fmt}(\text{hi}_k(e)) \subseteq \A(\infty, exp + k, b)$
  and $\text{fmt}(\text{lo}_k(e)) \subseteq \A(\infty, exp, 2^{exp + k})$.
Since the decomposition is exact, the rewrite is semantics-preserving.
When $\text{fmt}(e)$ is a fixed-point format,
  each part needs fewer digits than $e$,
  so finitization can choose smaller available formats
  for the operations that consume them.



Since multiple lowerings are possible for $P_2$,
  choosing one means searching the space of possible implementations.
This process is amenable to automation:
  both format inference
  and the format containment checks are straightforward and fast,
  and the search space of possible
  roundings and formats is bounded by the number of
  formats and operations in the programming environment;
  we leave automating the strategy to future work.
In \Cref{sec:case_study},
  we manually apply this strategy to
  the MX dot product example from \Cref{sec:overview} to
  develop optimized numerical models of the design.

\section{Evaluation}

\label{sec:eval}

We prototype a number library, dubbed \name, based on
  the approach outlined in \Cref{ssec:round:design}.
It is a \emph{correctly-rounded}, \emph{multi-precision}
  number library implemented in C++.

To evaluate the contributions of this paper, we pose the following
  research questions:

\begin{enumerate}[label*=\textbf{RQ\arabic*}]
\item Does \name support a broad class
  of number formats and rounding modes? (\Cref{ssec:eval:impl})
\item Is \name's performance competitive with
  comparable number libraries? (\Cref{ssec:eval:eval})
\item Can we use format inference and double rounding to safely implement
  a real-world algorithm, and can \name make the implementation more performant?
  (\Cref{sec:case_study})
\end{enumerate}

\Cref{ssec:eval:impl} describes \name's implementation and expressivity;
  \Cref{ssec:eval:eval} compares the performance of its
  correctly-rounded operations against SoftFloat, FloppyFloat, MPFR, and CPFloat;
  \Cref{sec:case_study} demonstrates how our techniques
  can be used to implement the MX dot product in a non-standard format,
  and how \name's primitives improve performance.

\subsection{Implementation}
\label{ssec:eval:impl}


\name implements optimized rounding procedures
  from natively-supported floating-point and integer types
  to instances of $\A$ under all seven rounding modes
  described in \Cref{ssec:bg:rounding}.
It computes a higher-precision, round-to-odd intermediate
  and then re-rounds it to the target.
This is correct whenever the intermediate format that
  \Cref{fig:correct-double-round} requires is contained in
  the type that holds it (binary64 or binary32);
  RNA follows from \textsc{rnd-RTO-RNE}, which holds for any tie-break,
  and RTP and RTN reduce to RAZ and RTZ by sign.

Any backend that produces this intermediate suffices;
  to demonstrate this generality, \name provides four:
  native floating point, the SoftFloat and FloppyFloat libraries,
  and custom kernels based on error-free transformations.
The current implementation pads every target uniformly by two bits,
  so in binary64\footnote{%
    In binary32, it supports $p \leq 22$ and $exp \geq -147$.}
  it supports every mode for
  $p \leq 51$ and $exp \geq -1072$;
  a debug-build assertion checks the precision,
  while the exponent and bound are preconditions.
The final rounding is a single software kernel for every mode:
  it splits the intermediate $x$ into $s + t$,
  where $s$ holds the target's retained bits,
  and returns $s$ or $\text{next}(s)$ depending on $t$ and the mode.
On overflow it uses IEEE 754's $\omega$,
  or saturates ($\omega(\pm 1) = \pm b$) for formats without infinities.

To ensure correctness,
  we test each backend exhaustively against MPFR,
  the de facto reference for correctly-rounded arithmetic,
  on every IEEE 754-style binary format of at most eight bits
  with a 2- to 5-bit exponent field
  (14 formats, each with subnormals, signed zeros, and infinities).
For each format, operation
  (\texttt{add}, \texttt{sub}, \texttt{mul}, \texttt{div}, \texttt{sqrt}, \texttt{fma}),
  rounding mode, backend, and container (\texttt{double} or \texttt{float}),
  we evaluate every input---%
  4,704 configurations and about 3.17 billion cases in total---%
  and compare results bit-for-bit, including the sign of zero\footnote{
    \name also implements IEEE 754-style exception flags;
    we verified their correctness against their IEEE 754 definitions.}.
MPFR lacks RTO and supports RNA only partially, so we derive both:
  RTO rounds under RTZ and forces the last bit odd if inexact;
  RNA equals RNE except at a midpoint, where it equals RAZ.
A randomized suite further checks unbounded-exponent formats
  of precision 2 to 8 against MPFR under five modes (all but RNA and RTO),
  on one million uniform inputs in $[-1, 1]$ per operation and mode
  ($[0, 1]$ for square root).

\emph{Native Floating Point}.
We emulate RTO arithmetic on top of hardware floating-point
  arithmetic, following \citet{boldo-rto}.
Since hardware does not natively support RTO, 
  we emulate it in hardware by
  performing the operation under RTZ
  and then adjusting the least significant
  bit of the significand.
We set the LSB to the logical OR of itself
  and the \emph{inexact flag}, which can
  be fetched from the architectural state.
After each operation,
  we restore the original rounding mode.
This approach is slow:
  it must set and fetch the floating-point
  control register twice per operation.
The resulting performance degradation
  has been previously documented~\cite{riscv-on-x64, rm-slow}.

\emph{SoftFloat}.
We extend the SoftFloat library using
  its double-precision arithmetic operations.
SoftFloat supports RTO,
  so each operation requires only setting
  the rounding mode (a global variable)
  before invoking the corresponding library function.

\emph{FloppyFloat}.
FloppyFloat provides similar capabilities to SoftFloat,
  but does not support RTO.
We emulate RTO as we did for hardware floating point:
  performing the operations in RTZ mode, then
  extracting the inexact flag provided by the library.

\emph{Error-free transformations}.
An error-free transformation (EFT) of an operation
  computes the rounded result and error term(s)
  such that their sum is the exact real result.
For example,
  the TwoSum algorithm~\cite{handbook}
  transforms $x + y$ into $s + t$ where
  $s$ is the nearest rounding and $t$ is the error term\footnote{%
  \citet{shewchuk} would classify $s$ and $t$
    as \emph{nonoverlapping}:
    the lowest set bit of $s$ is above the highest set bit of $t$.
  }.
We can transform the result of an EFT
  into the RTO result.
If $t$ is zero, then $s$ is the RTO result.
If $s$ and $t$ have the same sign,
  the RTO result is $s$ with the least significant bit
  of its significand set to 1.
Otherwise, we step $s$ one representable value toward zero,
  producing $s'$ (the error term now has the same sign as $s'$),
  and set the least significant bit of the significand of $s'$ to 1.

EFT algorithms~\cite{eftsanitizer} exist
  for multiplication and fused-multiply-add (FMA).
For division and square-root,
  we cannot represent the error term exactly;
  however, only the sign of the error term matters,
  so we can disregard any lower-order errors.
EFTs also extend to sums of $n$ numbers~\cite{shewchuk},
  so \name provides \texttt{add3} and \texttt{add4} operations
  that are rounded to odd.

This answers \textbf{RQ1}: in just \ImplSLOC source lines of code,
  \name's operations are correctly rounded for
  any instance of $\A$ whose \Cref{fig:correct-double-round} intermediate
  is contained in FP64, under all seven rounding modes.
Neither the arithmetic nor its round-to-odd emulation depends on
  the target format; only the final re-rounding does.
\name supports only the arithmetic operations
  required by IEEE 754,
  but could extend RTO-based math libraries
  such as RLIBM~\cite{rlibm} to support
  transcendental operations.

\subsection{Performance of MPFX}
\label{ssec:eval:eval}

We evaluate the performance of \name's correctly-rounded operations
  against existing number libraries that also guarantee correct rounding:
  MPFR, SoftFloat, and FloppyFloat.
We also compare against CPFloat~\cite{cpfloat},
  which differs from these libraries in two ways.
First, its interface operates on arrays rather than scalar values;
  to match the others, we call it with one element per operation
  and discuss its array interface separately.
Second, it re-rounds a round-to-nearest intermediate,
  so it is correctly rounded only where double rounding is innocuous:
  for round-to-nearest targets and the five IEEE 754 operations
  with $p \leq p_h/2 - 1$, at most 25 bits in binary64~\cite{figueroa, roux-double-round, rump-double-round}.
Its RTP and RTZ results in \Cref{fig:eval:overhead} carry no such guarantee.

Of \name's four arithmetic backends (\Cref{ssec:eval:impl}),
  the native-floating-point and EFT backends are the fastest;
  we evaluate these two and denote them \name/HW and \name/EFT.
The native-floating-point backend is slowed by the control-register
  manipulation its round-to-odd emulation requires, but its hardware
  arithmetic still makes it faster than \name's software-library backends.

We performed all experiments on \EvalOS with
  an \EvalCPU CPU and \EvalMem GB of RAM.
We compiled all libraries with \EvalCC at \texttt{-O2};
  MPFR \EvalMPFR was available as a package in \EvalOS.
We used SoftFloat \EvalSoftfl and
  the most recent version of FloppyFloat available\footnote{%
  \url{https://github.com/not-chciken/FloppyFloat/commit/124f2a3636cce89ffe59ac3cca1406e1740507c9}.
  }.
We used CPFloat \EvalCPFloat.

\emph{Methodology}.
We restrict to the shared features of each number library.
First,
  we focus on IEEE 754 required arithmetic:
  addition, subtraction, multiplication, division,
  square root, and FMA.
This eliminates transcendental functions
  (e.g., $\exp(x)$ and $\sin(x)$), and 
  algebraic operations (e.g., $\sqrt[3]{x}$),
  which only MPFR provides.
Second,
  these operations are \emph{uniform-precision}---that is,
  the arguments and result are the same number format---as
  only MPFR and \name are multi-precision.
We evaluate single-precision (FP32) and half-precision (FP16),
  the standard formats that our baselines support and that are narrow
  enough for \name to target by re-rounding from its FP64 working precision.
FloppyFloat relies on FP16 support
  introduced in the C++23 standard~\cite{cpp23}.
We exclude this treatment since on our machine,
  GCC\footnote{%
  Clang does not support FP16,
    so it fails to even compile FloppyFloat.
  } compiles FP16 operations to its own software
  floating-point implementation, as most x86 CPUs
  lack hardware half-precision arithmetic.
Last,
  we compare over three rounding modes:
  RNE, RTP, and RTZ.

For each operation and rounding mode,
  we sampled \EvalCmpNumInputs inputs from
  the uniform distribution on $[-1, 1]$ for each argument.
We then measured the execution time
  of evaluating over all the inputs in
  a tight loop for each of the number libraries,
  averaging over \EvalCmpNumIters iterations
  to reduce variance.
Finally,
  we computed the relative speedup over
  the corresponding SoftFloat implementation.
We take SoftFloat as the baseline because it is the standard,
  portable software implementation of correctly-rounded
  IEEE 754 arithmetic, with reasonable performance.

\begin{figure}[ht!]
\includegraphics[width=\textwidth]{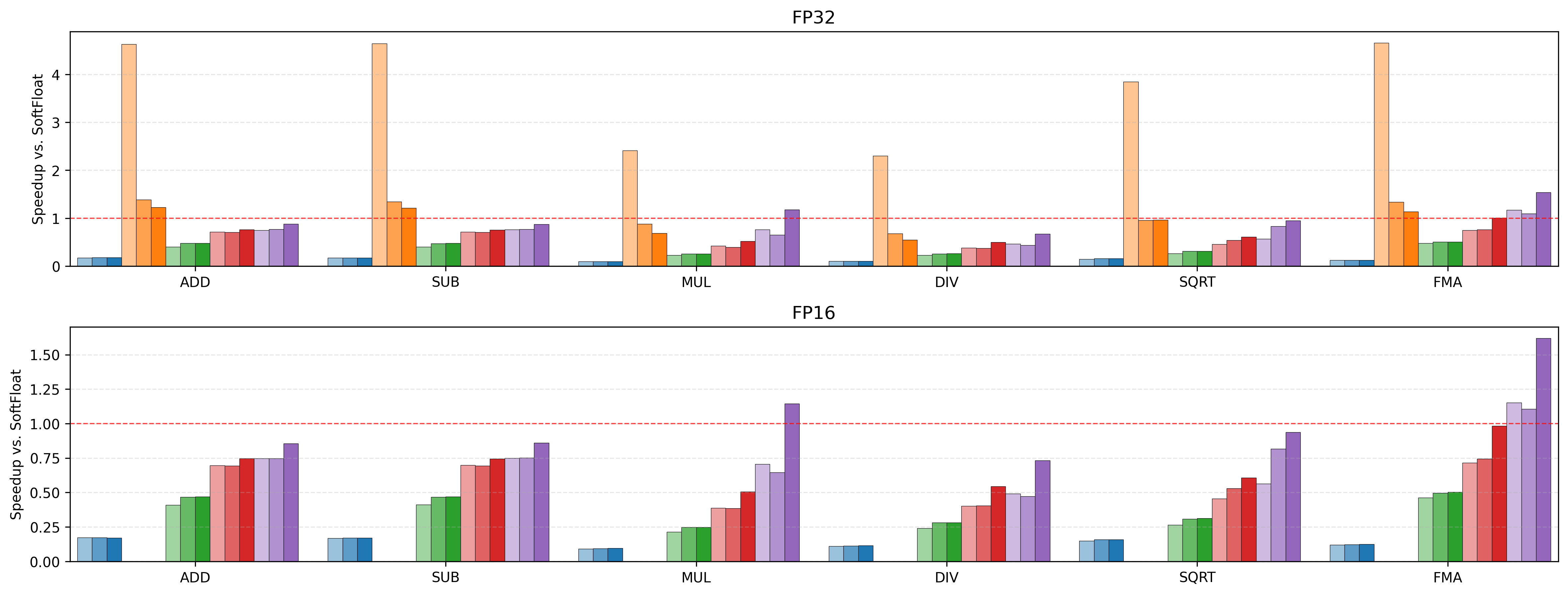}
\caption{%
Relative speedup over SoftFloat for each operation, datatype,
  and rounding mode.
The top and bottom panels show single-precision (FP32) and
  half-precision (FP16), with the six operations along the $x$-axis.
Within each operation, the bars are, left to right,
  MPFR, FloppyFloat (FP32 only), CPFloat, \name/HW, and \name/EFT, each as
  three bars for the RNE, RTP, and RTZ modes; the dashed line at
  $y=1$ is the SoftFloat baseline.
}
\label{fig:eval:overhead}
\end{figure}

\emph{Results}.
\Cref{fig:eval:overhead} shows the results of the comparison.
Being arbitrary-precision,
  MPFR is much slower than the other libraries,
  since managing variable-width significands adds overhead to every operation.
FloppyFloat is faster than SoftFloat
  for single-precision operations (FP32),
  especially for nearest rounding, since it
  uses floating-point rather than integer operations
  to model IEEE 754 arithmetic.
Our results resemble those reported by \citet{floppyfloat},
  but FloppyFloat is slower than SoftFloat for
  certain operators under RTP and RTZ.
We exclude FloppyFloat from the half-precision comparison
  for the reason given in the methodology.
With one element per call,
  CPFloat is slower than \name on every operation,
  by $1.6\times$ to $4.7\times$.
Called on arrays of \EvalCmpNumInputs inputs, it is instead
  $4.6\times$ to $12.2\times$ faster than \name,
  because it derives the target format's parameters
  (subnormal thresholds, overflow bound, significand masks)
  once per call rather than once per element.
The difference reflects API granularity, not the cost of rounding;
  a batched \name interface could amortize the same work.

In answer to \textbf{RQ2},
  \name is competitive with SoftFloat 
  despite being a compact library that is not highly optimized;
  across the operations, datatypes, and rounding modes we measure,
  the faster of its native-floating-point and EFT implementations runs between
  $\EvalMinSpeedup$ and $\EvalMaxSpeedup$ as fast as SoftFloat's,
  and up to $\SpeedupOverMPFR$ faster than MPFR.
The slowest of these is \name/EFT on FP32 division under RTP,
  and the fastest is \name/EFT on FP16 FMA under RTZ.
Because \name's native floating-point arithmetic (\name/HW)
  relies on modifying architectural state,
  we expect the results for that treatment to vary
  the most across different hardware architectures
  and generations.
This comparison uses only formats the baselines support;
  \Cref{sec:case_study} covers a non-standard one.

\section{Case Study: MX Dot Product}
\label{sec:case_study}


In \Cref{sec:overview},
  we motivated the need to implement
  the MX dot product algorithm using existing software operations;
  in \Cref{sec:analysis}, we introduced a strategy for mapping
  unavailable operations onto available arithmetic.
We now apply that strategy to construct a faithful numerical model
  of the MX dot product.
The program in \Cref{fig:overview:mx_spec} (left) is \emph{fully-specified}:
  every input yields a unique result.
We assume each block's 32 elements are finite values of
  the 8-bit MX E5M2 format, i.e., $\A(3, -16, 57344)$;
  following the standard~\cite{ocp-mx},
  each scale factor is a power of two in $[2^{-127}, 2^{127}]$.
The local dot product for each block pair (lines 6--9) is computed exactly;
  accumulation across pairs (line 10) is rounded to FP32 under RNE
  after each step.

\begin{wrapfigure}{r}{0.45\textwidth}
\centering
\vspace{-\baselineskip} 
\includegraphics[width=0.43\textwidth]{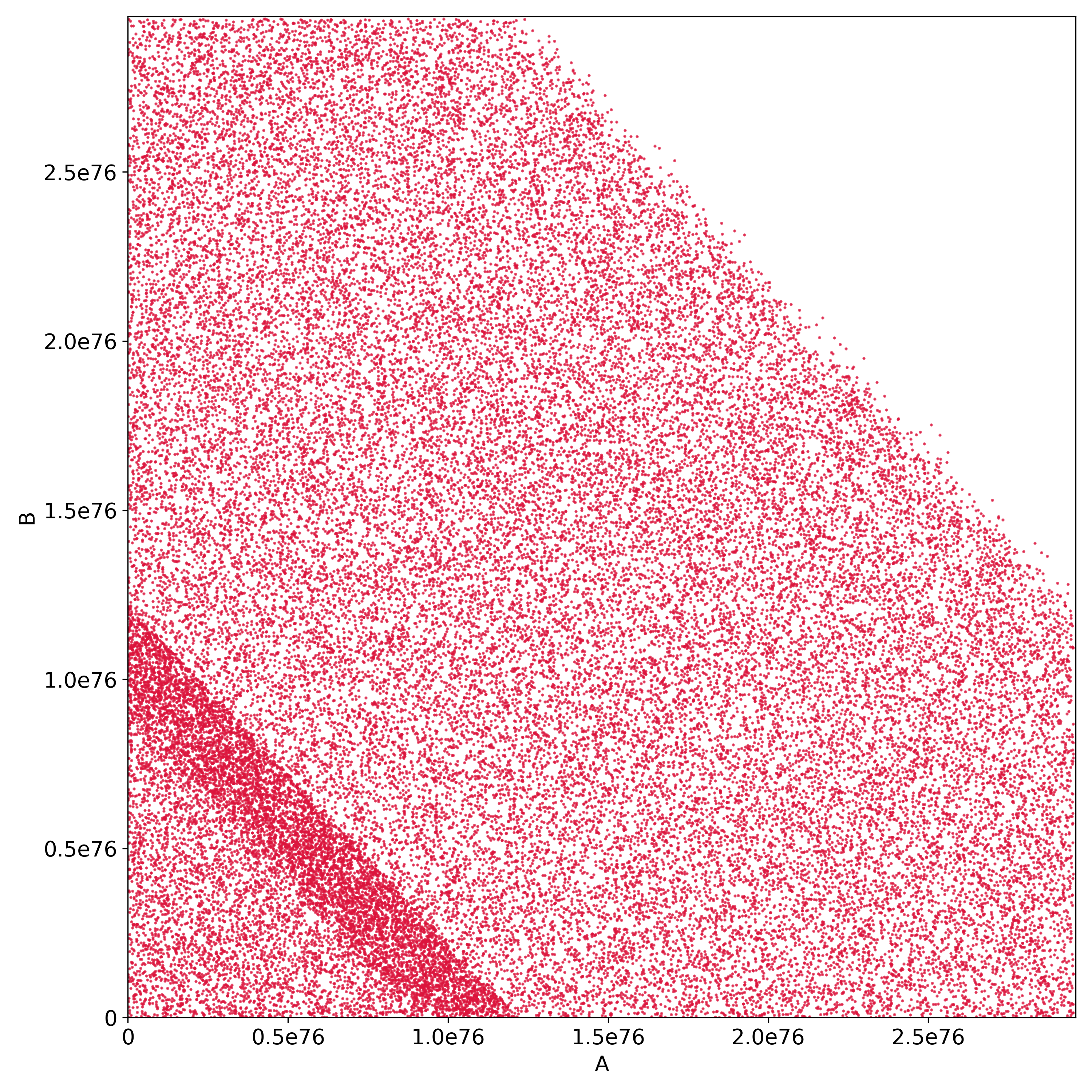}
\caption{%
Mismatches between the specification in \Cref{fig:overview:mx_spec} (left)
  and a naive implementation that runs the entire dot product uniformly in FP32.
}
\label{fig:eval:mismatch}
\vspace{-\baselineskip} 
\end{wrapfigure}

One might ask whether faithfully modeling the specification is necessary at all---%
  could a naive implementation, running the entire dot product uniformly in FP32, suffice?
\Cref{fig:eval:mismatch} answers this question:
  we ran this naive implementation on a million randomly sampled MX block pairs
  (with $C = 0$) and plotted each disagreement as a point,
  using the encodings of blocks $A$ and $B$ as coordinates\footnote{%
  Each block encodes to a 264-bit integer: 8 bits for the scale factor
    and 256 bits for the 32 elements.
  }.
The implementation produces incorrect results
  on roughly \CaseStudyMismatches of the inputs.
The absence of error in the upper right corner of
  the plot corresponds to overflow and the dense
  band in the lower left corresponds to underflow;
  otherwise, the errors are distributed uniformly across the input space.
FP32 is simply insufficient to faithfully model the specification,
  motivating the more careful approach of \Cref{sec:analysis}.
%
%
%
Next,
  we show how to properly map the algorithm onto SoftFloat's
  IEEE 754 operations, then highlight how \name's richer primitives
  yield a much faster implementation.

\subsubsection*{Mapping to SoftFloat}

At first glance,
  it seems unclear how to map each operation of
  the program to rounded operations supported
  by SoftFloat (or native hardware).
In addition,
  we do not know what number formats can safely store
  any intermediate computations.
As an initial step,
  we apply \emph{format inference} to
  bound each expression by a number format.
\Cref{fig:eval:mx_dot_prod_inferred} shows the result.

\begin{figure}[ht!]
\begin{minted}[escapeinside=||]{Python}
def accurate_mx_dot(|$A$|, |$B$|):
  |$C$| = 0
  for (|$X^{A}_i$|, |$P^{A}_i$|), (|$X^{B}_i$|, |$P^{B}_i$|) in zip(|$A$|, |$B$|):
    |$X$| = |$\text{mul}_{\R,\_}(X^{A}_i, X^{B}_i)$|  |{\scriptsize\textcolor{RoyalBlue}{$\A(1, -127, 2^{127}) \otimes \A(1, -127, 2^{127}) = \A(2, -254, 2^{254})$}}|
    |$V$| = 0
    for |$P^{A}_{i,j}$|, |$P^{B}_{i,j}$| in zip(|$P^{A}_i$|, |$P^{B}_i$|):
      |$t_1$| = |$\text{mul}_{\R,\_}(P^{A}_{i,j}, P^{B}_{i,j})$|  |{\scriptsize\textcolor{RoyalBlue}{$\A(3, -16, 7 \times 2^{13}) \otimes \A(3, -16, 7 \times 2^{13}) = \A(6, -32, 49 \times 2^{26})$}}|
      |$V$| = |$\text{add}_{\R,\_}(V, t_1)$|  |{\scriptsize\textcolor{RoyalBlue}{$\bigoplus_{j=1}^{32} \A(6, -32, 49 \times 2^{26}) = \A(69, -32, 1568 \times 2^{26})$}}|
   |$t_2$| = |$\text{mul}_{\R,\_}(X, V)$|  |{\scriptsize\textcolor{RoyalBlue}{$\A(2, -254, 2^{254}) \otimes \A(69, -32, 1568 \times 2^{26}) = \A(71, -286, 1568 \times 2^{280})$}}|
   |$C$| = |$\text{add}_{\text{FP32},\text{RNE}}(C, t_2)$|  |{\scriptsize\textcolor{RoyalBlue}{$\rightarrow \A(24, -149, (2^{24}-1) \times 2^{104})$ (FP32)}}|
  return |$C$|
\end{minted}
\caption{%
The specification of \Cref{fig:overview:mx_spec} (left),
  annotated in blue with the format inferred for each intermediate:
  operand formats combine via $\otimes$ (multiplication) or
  $\oplus$ (accumulation).
The inner loop is unrolled for the analysis,
  so line~8's $\oplus$ spans all $32$ products.
Line~10's FP32 format is set by the specification, not inferred.
}
\label{fig:eval:mx_dot_prod_inferred}
\end{figure}

Each scale factor lies in $\A(1, -127, 2^{127})$
  and each block element in the finite values of E5M2, $\A(3, -16, 7 \times 2^{13})$;
  format inference propagates these bounds through the program
  to derive the rest\footnote{%
  For the scale-factor product, the resulting format could have precision
    $p = 1$ rather than $p = 2$.}.
The accumulator $V$ is the interesting case:
  its precision reaches $p = 69$, the precision of its largest value
  $1568 \times 2^{26} = c \times 2^{-32}$ (where $1568 = 32 \times 49$)
  under the quantum $2^{-32}$---already beyond double precision ($p = 53$).
After scaling by $X$, the local dot product needs
  $\A(71, -286, 1568 \times 2^{280})$.

Recall that the SoftFloat library
  only supports IEEE 754 floating-point number formats---%
  for example, it has no (unrounded) multiplication implementation (line 7)
  with the signature
\[
  \A(3, -16, 7 \times 2^{13})
  \times \A(3, -16, 7 \times 2^{13})
  \rightarrow \A(6, -32, 49 \times 2^{26}).
\]
We can overcome this issue by finitization (\Cref{ssec:analysis:compile}),
  promoting to a supported format that contains the inferred one.
For the multiplication,
  both the input and resulting number formats
  satisfy $F \subseteq \text{FP32}$,
  so we can promote both input and output types
  to FP32:
\[
\text{mul}_{\text{FP32},\text{RNE}}(
  \rnd_{\text{FP32}, \text{RNE}}(P^{A}_{i,j}),
  \rnd_{\text{FP32}, \text{RNE}}(P^{B}_{i,j})).
\]
Here,
  $\rnd_{\text{FP32}, \text{RNE}}$ coerces $P^{A}_{i, j}$
  and $P^{B}_{i,j}$ into a single-precision floating-point number.
The conversion is lossless,
  so any rounding mode suffices; we choose RNE.
Now,
  the multiplication is a uniform,
  single-precision, floating-point multiplication
  which SoftFloat---and native hardware---provide.

For the rest
  of the program, we continue promoting the number formats
  produced by format inference.
Whenever the resulting number format of
  an operation is larger than the arguments,
  we promote the arguments to the larger type.
The result of the process is
  shown in \Cref{fig:eval:mx_dot_prod_sf}.

\begin{figure}[ht!]
\begin{minted}[escapeinside=||]{Python}
def our_mx_dot_softfloat(|$A$|, |$B$|):
  |$C$| = 0
  for (|$X^{A}_i$|, |$P^{A}_i$|), (|$X^{B}_i$|, |$P^{B}_i$|) in zip(|$A$|, |$B$|):
    |$X$| = |$\hwop{\text{mul}_{\text{FP64},\text{RNE}}}(\hwop{\text{rnd}_{\text{FP64},\text{RNE}}}(X^{A}_i), \hwop{\text{rnd}_{\text{FP64},\text{RNE}}}(X^{B}_i))$|
    |$V$| = 0
    for |$P^{A}_{i,j}$|, |$P^{B}_{i,j}$| in zip(|$P^{A}_i$|, |$P^{B}_i$|):
      |$t$| = |$\hwop{\text{mul}_{\text{FP32},\text{RNE}}}(\hwop{\text{rnd}_{\text{FP32},\text{RNE}}}(P^{A}_{i,j}), \hwop{\text{rnd}_{\text{FP32},\text{RNE}}}(P^{B}_{i,j}))$|
      |$V$| = |$\nlop{\text{add}_{\text{FP128},\text{RNE}}}(V, \nlop{\text{rnd}_{\text{FP128},\text{RNE}}}(t))$|
    |$t_2$| = |$\nlop{\text{mul}_{\text{FP128},\text{RNE}}}(X, V)$|
    |$t_3$| = |$\nlop{\text{add}_{\text{FP128},\text{RTO}}}(t_2, \nlop{\text{rnd}_{\text{FP128},\text{RNE}}}(C))$|
    |$C$| = |$\nlop{\text{rnd}_{\text{FP32},\text{RNE}}}(t_3)$|
  return |$C$|
\end{minted}
\caption{%
Our implementation of the dot product
  using \nlop{SoftFloat} or \hwop{hardware-supported}
  floating-point operations.
This program is semantically equivalent 
  to the original specification.
}
\label{fig:eval:mx_dot_prod_sf}
\end{figure}

The program uses a mix of number formats,
  but every operation is uniform-precision,
  with added casts where necessary.
Some operations can be supported in hardware,
  such as the double-precision multiplication
  of the scale factors, and the single-precision
  multiplication of the elements.
To support the largest number format,
  $\A(71, -286, 1568 \times 2^{280})$,
  we use FP128, or quad-precision floating-point,
  from SoftFloat, which is a very large number format 
  whose operations are slower than those of smaller formats.
We implement the final rounded addition
  using a correct double rounding,
  first rounding to FP128,
  and then re-rounding to FP32 to update $C$.

Our approach is essential for implementing
  the dot product faithfully.
Format inference identifies
  number formats for every expression,
  and helps rewrite the program to only use supported operations
  by promoting number formats where necessary.
The result of this process is a program
  that matches the intended numerical behavior
  without requiring the developer to create
  bespoke implementations of arithmetic operations.
However,
  a limitation with a library like SoftFloat
  is that it provides only a small set of primitives,
  so some operations must use larger-than-necessary types,
  which can incur a performance penalty.

\subsubsection*{Going Faster with \name}

Compared to SoftFloat,
  \name offers a very large set of primitives,
  composed of different arithmetic methods and
  re-rounding to a variety of number formats.
For users,
  this flexibility allows them to choose
  the arithmetic method and target format,
  making it easier to implement correctly-rounded operations
  and to express the number format they actually want.
Does this flexibility lead to a more performant dot product
  implementation?

We target the slowest part of the SoftFloat design:
  the accumulation in FP128.
Recall that format inference derives
  $\A(69, -32, 1568 \times 2^{26})$
  as the number format for $V$.
Rather than selecting FP128,
  we choose a fixed-point format, $\A(\infty, -32, 2^{95} - 2^{-32})$,
  which can be encoded in a 128-bit integer with
  an implicit scaling factor of $2^{-32}$ (most C/C++ compilers support 128-bit integers through
  non-standard extensions\footnote{
  GCC and Clang call the 128-bit integer type \texttt{\_\_int128}.
  }).
The accumulation can now be performed
  using standard integer arithmetic, which is
  much faster than SoftFloat's addition operation.

Unfortunately,
  the scaling factor complicates our plan.  
Post-scaling,
  the local dot product is representable
  in $\A(71, -286, 1568 \times 2^{280})$;
  a valid fixed-point format for this value
  would need to be backed by a 1024-bit integer,
  at minimum; the scale factor, $X$, would require
  a 512-bit integer.
The workaround is to split $V$ (\Cref{ssec:analysis:compile})
  into $V_{hi} + V_{lo}$ with $k = 40$:
  since only 69 digits of the 128-bit integer backing $V$
  are significant, each part fits a 64-bit integer,
  $V_{hi}$ with scale factor $2^{-32+40}$
  and $V_{lo}$ with scale factor $2^{-32}$.

This decomposition allows us to rewrite
  $(X \times_{\R} V) +_{\text{FP32}, \text{RNE}} C$
  into a three-argument addition:
\[
\texttt{add3}_{\text{FP32},\text{RNE}}(
    \text{mul}_{\R,\_}(X, V_{hi}),
    \text{mul}_{\R,\_}(X, V_{lo}),
    C
)
\]
The remaining question is how to
  implement the two new factors
  $X \times_{\R} V_{hi}$ and $X \times_{\R} V_{lo}$.
If we split with 40 digits,
  we can infer that $V_{hi}$ must be contained in
  $\A(29, 8, 1568 \times 2^{26})$
  and $V_{lo}$ is contained in
  $\A(40, -32, 2^{8})$.
Then $X \times_{\R} V_{hi}$ and $X \times_{\R} V_{lo}$
  are representable in $\A(31, -246, 1568 \times 2^{280})$
  and $\A(42, -286, 2^{262})$, respectively.
Crucially,
  all of these number formats are contained
  in a double-precision floating-point format.
Therefore,
  we can convert $V_{hi}$ and $V_{lo}$
  to double-precision floating-point values
  and perform the multiplication using
  double-precision arithmetic uniformly.

\begin{figure}[ht!]
\begin{minted}[escapeinside=||]{Python}
def our_mx_dot_mpfx(|$A$|, |$B$|):
  |$C$| = 0
  for (|$X^{A}_i$|, |$P^{A}_i$|), (|$X^{B}_i$|, |$P^{B}_i$|) in zip(|$A$|, |$B$|):
    |$X$| = |$\hwop{\text{mul}_{\text{FP64},\text{RNE}}}(\hwop{\text{rnd}_{\text{FP64},\text{RNE}}}(X^{A}_i), \hwop{\text{rnd}_{\text{FP64},\text{RNE}}}(X^{B}_i))$|
    |$V$| = 0
    for |$P^{A}_{i,j}$|, |$P^{B}_{i,j}$| in zip(|$P^{A}_i$|, |$P^{B}_i$|):
      |$t$| = |$\hwop{\text{mul}_{\text{FP32},\text{RNE}}}(\hwop{\text{rnd}_{\text{FP32},\text{RNE}}}(P^{A}_{i,j}), \hwop{\text{rnd}_{\text{FP32},\text{RNE}}}(P^{B}_{i,j}))$|
      |$V$| = |$\hwop{\text{add}_{\A(\infty, -32, 2^{95} - 2^{-32}),\text{RNE}}}(V, \nlop{\text{rnd}_{\A(\infty, -32, 2^{95} - 2^{-32}),\text{RNE}}}(t))$|
    |$V_{hi}, V_{lo}$| = |$\nlop{\text{split}}(V, 40)$|
    |$t_2$| = |$\hwop{\text{mul}_{\text{FP64},\text{RNE}}}(X, \nlop{\text{rnd}_{\text{FP64},\text{RNE}}}(V_{hi}))$|
    |$t_3$| = |$\hwop{\text{mul}_{\text{FP64},\text{RNE}}}(X, \nlop{\text{rnd}_{\text{FP64},\text{RNE}}}(V_{lo}))$|
    |$C$| = |$\nlop{\text{add3}_{\text{FP32},\text{RNE}}}(t_2, t_3, C)$|
  return |$C$|
\end{minted}
\caption{%
Our implementation of the dot product
  using \nlop{\name} or \hwop{hardware-supported}
  floating-point operations.
This program is semantically equivalent 
  to the original specification.
}
\label{fig:eval:mx_dot_prod_ours}
\end{figure}

The full implementation is shown
  in \Cref{fig:eval:mx_dot_prod_ours}
  using a mix of native integer and floating-point operations
  and \name's primitives.
By replacing FP128 accumulation with                                                                                                                            
  128-bit integer arithmetic
  and using native FP64 for the remaining multiplications,
  the implementation avoids SoftFloat's
  most expensive operations entirely.
The split and three-argument adder from \name
  ensure the final rounding is still correct.

\subsubsection*{Comparison}

We repeated the same approach, choosing
  two additional combinations of number formats
  for the element vectors $P^{A}_{i}$ and $P^{B}_{i}$:
  $\text{E5M2}\times\text{E4M3}$,
  and $\text{E4M3}\times\text{E4M3}$.
For each configuration,
  we developed a SoftFloat implementation and an \name implementation
  using the previously-described techniques.
We sampled 100,000 input vectors of length 4096
  from a standard normal distribution and
  applied block quantization~\cite{ocp-mx} to each vector
  to form MX blocks.
We then measured the run time of each program
  on these MX blocks and computed
  the relative speedup of the \name implementation
  over the SoftFloat implementation.

\begin{table}[ht!]
\caption{%
Relative speedup of \name's implementations compared to SoftFloat's implementations
  for three different combinations of number formats for the element vectors.
}
\begin{tabular}{l|l|l|l}
\textbf{Formats} &
  $\text{E5M2}\times\text{E5M2}$ &
  $\text{E5M2}\times\text{E4M3}$ &
  $\text{E4M3}\times\text{E4M3}$ \\
\hline
\textbf{Speedup} & $9.99\times$ & $13.26\times$ & $1.40\times$
\end{tabular}
\label{fig:eval:speedup}
\end{table}

\Cref{fig:eval:speedup} shows our results.
For the worked example,
  labeled $\text{E5M2}\times\text{E5M2}$ in the table,
  we achieve a $9.99\times$ speedup over our
  SoftFloat implementation.
For $\text{E5M2}\times\text{E4M3}$,
  the \name implementation uses 64-bit integers
  instead of 128-bit integers,
  but the overall approach is similar,
  including splitting the accumulator 
  and the three-argument adder;
  we achieve a greater speedup of $13.26\times$.
For $\text{E4M3}\times\text{E4M3}$,
  format inference proves that we can
  use double-precision arithmetic uniformly
  throughout the inner dot product, including scaling.
This observation applies to both
  the SoftFloat and \name implementation.
The speedup for this configuration is smaller:
  the only difference between the programs
  is the implementation of the final adder.

\subsubsection*{Discussion}
In this case study, we show how the techniques from \Cref{sec:analysis}
  can be applied to a real-world algorithm
  using existing library operations. 
Without these techniques, it would be difficult to faithfully implement 
  the algorithm over a non-standard format, answering the first part of 
  \textbf{RQ3}.
Further,
  \name's support for non-standard formats and multi-precision
  operations allows us to exploit cheaper arithmetic, 
  which is inexpressible in SoftFloat, answering the second part of
  our research question.
Underlying both is the abstract number format $\A$:
  format inference uses it to determine
  what each expression requires,
  and \name uses it to round to exactly those formats.

\section{Limitations and Future Work}
\label{sec:future}


We would like to support additional number formats.
Posit numbers~\cite{posit} use a dynamic rather than a fixed precision
  that varies with the magnitude of the representable value.
Strictly positive formats, such as the MX E8M0 format~\cite{ocp-mx}
  used for block scales, are not directly representable in $\A$.
The most recent draft of the IEEE P3109 standard~\cite{p3109}
  introduces unsigned floating-point numbers;
  we believe setting $b^{-} = 0$ would suffice to support them,
  but have not yet verified this.


%

Our formats have a single zero.
IEEE 754 distinguishes $+0$ and $-0$;
  modeling signed zero would let our rules account for its sign.
Subnormals pose a challenge:
  floating-point implementations often allow disabling subnormalization,
  flushing subnormal values to zero.
We would like to investigate how flushing affects
  format containment and correct double rounding.

Format inference (\Cref{ssec:analysis:infer}) covers straight-line programs;
  loops with a static bound can be unrolled,
  as in our case study (\Cref{sec:case_study}).
Supporting loops with dynamic bounds requires,
  as in interval analysis, a widening operator,
  since formats have infinite ascending chains under containment;
  the resulting formats may be impractically wide.

Finally,
  we plan to improve \name.
Since it supports single-precision containers,
  we can integrate the RLIBM~\cite{rlibm} math library
  to support transcendental functions:
  its operations are rounded to odd
  and can be re-rounded to single precision.
We also plan to adopt CPFloat's optimized rounding algorithms~\cite{cpfloat},
  which target exactly \name's rounding primitives.


\section{Related Work}
\label{sec:related}

Our work focuses on number libraries,
  floating-point formalizations,
  double rounding, and rounding to odd.
We discuss related work in each of these areas.


\emph{Number libraries}.
\name is most similar to CPFloat~\cite{cpfloat};
  CPFloat uses hardware floating-point arithmetic for higher precision
  computation and then optimized rounding kernels to round to the target format.
In particular,
  they support IEEE 754-like number formats:
  $\A(p, e_{\text{min}} - p + 1, (2 - 2^{1 - p}) \times 2^{e_{\text{max}}})$,
  and also offer stochastic rounding, which \name lacks.
The focus of the work concerns rounding algorithms:
  for basic arithmetic operations,
  they appeal to strong assumptions for correct rounding
  found in prior work~\cite{rump-double-round,roux-double-round}.
MXSim~\cite{mxsim} builds on CPFloat to simulate MX-like
  block floating-point formats with customizable
  element, scale, and accumulator formats,
  and requires the same assumptions for correct rounding.
FloatX~\cite{floatx} employs a similar technique,
  but stores results in a custom container type.
They note that their double roundings
  may not be correctly rounded.

Most number libraries do not follow
  the double rounding architecture of \name.
The SoftFloat library~\cite{softfloat} emulates
  IEEE 754 floating-point arithmetic using only machine integer operations,
  providing highly optimized implementations of basic arithmetic operations.
FloppyFloat~\cite{floppyfloat} supports a subset
  of SoftFloat's operations.
The library instead computes the rounded result under RNE
  using hardware floating-point arithmetic, and fixes up
  the result for the desired rounding mode.
The SIPE library~\cite{sipe} represents numbers
  as a pair of integers (significand and exponent) and
  simulates efficient, low-precision, floating-point arithmetic;
  underflow and overflow are not modeled.
Like SoftFloat and FloppyFloat,
  the operations provided by SIPE are uniform-precision:
  the precision of the arguments and result is the same.
The Universal library~\cite{universal,universal2} is an extensive
  compendium of fixed-point, floating-point, posit, and logarithmic
  formats, among others.
Supporting them with the traditional approach is, as the developers note,
  a ``major undertaking with complex validation requirements''~\cite{universal}.

Finally,
  arbitrary-precision number libraries allow
  theoretically unbounded precision.
Because of their unbounded nature,
  their implementations must allocate memory
  to represent numbers.
The state-of-the-art
  arbitrary-precision floating-point library is MPFR~\cite{mpfr},
  which provides correctly-rounded operations for an array
  of arithmetic, algebraic, and transcendental functions.
Arbitrary-precision integer libraries, such as GNU MP~\cite{gmp},
  are arguably number libraries too.


\emph{Floating-point formalization}.
Generic formalizations of floating-point arithmetic
  predate our work.
The PFF library~\cite{pff} formalizes floating-point numbers in Coq~\cite{coq}
  over an arbitrary radix as (significand, exponent) pairs whose exponent
  is bounded below but \emph{unbounded above},
  deferring overflow to a separate, final step.
\citet{harrison} gives a comparable radix-two theory in HOL Light that
  carefully separates a number's real \emph{value} from its bit-string
  \emph{encoding}; like PFF, it rounds into a format unbounded above and
  detects overflow afterward, mirroring how we round into $\hat{F}$
  before applying the bound $b$ (\Cref{ssec:bg:rounding}).
Both presentations target error analysis and the verification of arithmetic
  and hardware rather than format containment or double rounding,
  and neither treats fixed-point as a first-class format.
\citet{jacobsen} extend Harrison's theory to an arbitrary radix and glue
  fixed- and floating-point numbers together to model IEEE 754 normals and
  subnormals; our format instead unifies them as one family
  $\A(p, exp, b)$, with fixed-point arising when $p = \infty$.

The Flocq library~\cite{flocq} is a more recent and comprehensive Coq
  formalization~\cite{coq} whose generic format---parameterized by a radix
  and a canonical-exponent function---likewise covers fixed- and
  floating-point formats with optional subnormalization.
Our abstract number format corresponds to this generic format
  with an added bound and special values (\Cref{ssec:format:abstract});
  we use its particular parameterization to decide format containment
  and correct double rounding.
More recently, \citet{flops} give a Lean formalization of the draft IEEE
  P3109 standard for low-precision machine-learning formats.
Like ours, their work is mechanized in Lean and must grapple with the
  degeneracies of one-bit-precision formats; in contrast, it aims to
  faithfully specify a single standard and its algorithms rather than
  reason about format containment or double rounding across formats.

\emph{Double rounding}.
The seminal work on correct double rounding is by \citet{boldo-rto}.
They prove that round to odd (RTO) and $p + 2$ digits of precision
  ensures correct rounding to $p$ digits under the RNE mode.
Like theirs,
  our rules place conditions on the formats alone,
  so they apply to any real number and thus any function;
  \Cref{ssec:round:safe} credits the prior work behind each rule.
\citet{figueroa} gives precision conditions under which RNE double rounding
  is innocuous for each basic arithmetic operation,
  roughly $2p$ digits for a $p$-digit target;
  \citet{roux-double-round} formalizes these conditions in Coq,
  including underflow and radices other than two, and
  \citet{rump-double-round} extends them to directed rounding
  and special values such as signed zero.
\citet{martin-dorel} study the converse problem:
  how unavoidable double rounding affects algorithms such as 2Sum.

\emph{Round to odd}.
\citet{boldo-rto-apps} present a few formally-verified
  applications of the round to odd rounding mode
  including a fused-multiply-add (FMA) algorithm
  and correctly-rounded summation.
The RLIBM math library~\cite{rlibm, rlibm-all} consists
  of automatically-generated polynomial approximations
  of transcendental functions that are rounded to odd,
  allowing correct double rounding to single precision
  under any of the standard rounding modes.
Our work applies similar techniques to number libraries,
  an approach mentioned but not explored in \citet{floatx}.
We were particularly inspired by RLIBM's philosophy
  of focusing on representable values rather than real numbers
  in our formalism and applications.

\emph{Error analyzers}.
A separate line of work builds tools that automatically bound the
  round-off error of floating-point computations.
We do not compute error bounds, but these tools share our reasoning about
  when an operation is exact.
VCFloat~\cite{vcfloat, vcfloat2} proves sound round-off error bounds in a
  proof assistant and recognizes exact operations, such as Sterbenz
  subtraction~\cite{handbook} and multiplication by powers of two, to tighten them.
Gappa~\cite{gappa} certifies bounds over both fixed- and floating-point
  rounded operators, detecting exact operations with the same two
  properties our format $\A(p, exp, b)$ unifies: that a value is an integer
  multiple of a quantum, and that its significand fits in $p$ bits.
These tools exploit exactness to tighten error bounds; we exploit it
  through format containment to drop or insert roundings while preserving
  correct rounding.
Type systems can also bound roundoff error:
  Numerical Fuzz~\cite{numfuzz} targets forward error,
  and Bean~\cite{bean}, backward error.

\emph{Rewriting numerical programs}.
Like our rewriting strategy (\Cref{sec:analysis}), verified compilers
  transform numerical programs to map operations onto those available in a
  programming environment.
RealCake~\cite{realcake} applies fast-math optimizations in CakeML,
  proving the optimized result stays within a user-specified error bound
  of the program's real-valued semantics rather than preserving its exact
  value.
Our transformations instead keep the correctly-rounded result exactly.

\emph{Precision tuning}.
Mixed-precision tuning also assigns number formats
  to a program's operations, but toward a different end.
Tools such as Precimonious~\cite{precimonious}, FPTuner~\cite{fptuner},
  Daisy~\cite{daisy}, and POPiX~\cite{popix} search for low-precision
  assignments that minimize cost while keeping the result within an accuracy
  bound, exploring programs with \emph{different} semantics;
  \citet{isychev-lohar} give a recent overview.
Our rewriting strategy (\Cref{sec:analysis}), by contrast, applies only
  \emph{semantics-preserving} rewrites, choosing formats that preserve
  the program's exact result rather than trading accuracy for cost.
The two are complementary:
  a precision tuner could choose a program's roundings,
  and our format containment could then check whether
  that program maps onto the available primitives,
  or decide when two candidate programs are equivalent
  to prune the search.

\section{Conclusion}
\label{sec:conclusion}

We characterize which pairings of rounding modes
  support correct double rounding,
  with easily checkable sufficient conditions on the formats,
  enabled by the abstract number format $\A$---%
  which parameterizes fixed- and floating-point formats---%
  and an efficient, complete check for format containment.
Our work complements prior results that apply
  to specific operations~\cite{rump-double-round, roux-double-round}
  or pairs of rounding modes~\cite{boldo-rto}.
Format containment allows us to reason about correct double rounding at two levels:
  for individual operations, it determines when double rounding preserves correctness;
  for whole programs, it drives a rewriting strategy that maps unsupported operations
  onto available library or hardware primitives.
\name demonstrates that a modular number library based on double rounding
  can be competitive with established libraries like MPFR and SoftFloat,
  despite supporting a far broader class of number formats.
A case study on a real-world machine learning algorithm shows
  that targeting \name's primitives yields faster implementations.
As hardware number formats continue to diversify,
  correct double rounding offers a principled and scalable path to supporting them in software.

\bibliographystyle{ACM-Reference-Format}
\bibliography{references}

\ifdefined\hideappendix\else
\clearpage
\appendix

\section{Proofs}
\label{sec:appendix:proof}

This appendix contains proofs for
  the format containment rules (\refRoundContain),
  the correct double rounding rules (\refRoundSafe),
  and the format inference rules (\refAnalysisInfer).
All of these proofs are mechanized in Lean;
  \Cref{fig:appendix:lean} maps rules and lemmas to their Lean theorems
  (in namespace \texttt{Mpfx}).

\begin{table}[ht]
\caption{Lean theorems for the principal rules and lemmas of Sections 5 and 6.}
\label{fig:appendix:lean}
\small
\begin{tabular}{ll}
\toprule
Result & Lean theorem \\
\midrule
\textsc{$\A$-Contains-Prec}, \textsc{$\A$-Contains-Sub}
  & \texttt{Format.containsPrec}, \texttt{Format.containsSub} \\
\refLemContainComplete & \texttt{Format.subset\_iff\_contains} \\
\refLemKeptDigits & \texttt{FiniteFormat.numDigits} (definition) \\
\refLemKeptCmp & \texttt{FiniteFormat.numDigits\_extend} \\
\refLemRtoRounding & \texttt{IsOdd.transfer\_of\_subset} \\
\textsc{rnd-RTZ-RTZ}, \textsc{rnd-RAZ-RAZ}
  & \texttt{roundsRTZ\_RTZ\_inBound}, \texttt{roundsRAZ\_RAZ\_inBound} \\
\textsc{rnd-RTO-RTO} & \texttt{roundsRTO\_RTO\_inBound} \\
\textsc{rnd-RTO-RTZ}, \textsc{rnd-RTO-RAZ}
  & \texttt{roundsRTO\_RTZ\_inBound}, \texttt{roundsRTO\_RAZ\_inBound} \\
\textsc{rnd-RTO-RNE} (any tie-break) & \texttt{roundsRTO\_RN\_inBound} \\
\quad (both overflow) & \texttt{rounds}$X$\texttt{\_}$Y$\texttt{\_noOverflow} (six theorems) \\
\textsc{rnd-Double} & \texttt{rnd\_double}; per rule, \texttt{rnd}$X$\texttt{\_}$Y$\texttt{\_of\_bound} \\
Counterexamples (\refTabDoubleRoundCex)
  & \texttt{Cex.no\_rnd}$X$\texttt{\_}$Y$ (ten theorems) \\
$\otimes$ and $\oplus$ rules
  & \texttt{Format.mul\_subset}, \texttt{Format.add\_subset} \\
Negation and absolute value
  & \texttt{Format.neg\_subset}, \texttt{Format.abs\_subset} \\
\bottomrule
\end{tabular}
\end{table}

\subsection{Format Containment}
\label{ssec:appendix:proof:contain}

We prove soundness of the two inference rules
  that derive $\A(p_1, exp_1, b_1, S_1) \subseteq \A(p_2, exp_2, b_2, S_2)$,
  and then their completeness (\refLemContainComplete).
Special values are not numbers,
  so this containment holds exactly when $S_1 \subseteq S_2$
  and the numerical values are contained;
  below we prove the latter, writing $\A(p, exp, b)$ for the numerical values.
\textsc{$\A$-Contains-Prec} handles the general case,
  while \textsc{$\A$-Contains-Sub} handles
  degenerate floating-point formats (and bounded fixed-point formats)
  where \textsc{$\A$-Contains-Prec} does not apply.



\begin{proof}[Proof of \textsc{$\A$-Contains-Prec}]
Let $x \in \A(p_1, exp_1, b_1)$.
The precision and magnitude constraints transfer
  to $\A(p_2, exp_2, b_2)$ directly:
  $x$ has precision at most $p_1 \leq p_2$,
  and $|x| \leq b_1 \leq b_2$.
It remains to check the quantum constraint, on which we case split.

If $exp_2 = -\infty$,
  then $\A(p_2, {-\infty}, b_2)$ imposes no quantum constraint,
  and the two facts above give $x \in \A(p_2, {-\infty}, b_2)$.
Otherwise $exp_2 > -\infty$,
  and since $exp_1 \geq exp_2$ we have $exp_1 > -\infty$ as well.
Write $x = (-1)^s \times c \times 2^{exp_1}$
  where $c$ has precision at most $p_1 \leq p_2$.
Since $exp_1 \geq exp_2$, we can rewrite
\[
  x = (-1)^s \times (c \times 2^{exp_1 - exp_2}) \times 2^{exp_2},
\]
so $x$ is an integer multiple of $2^{exp_2}$,
  satisfying the quantum constraint.
Therefore, $x \in \A(p_2, exp_2, \infty)$.
\end{proof}


\begin{proof}[Proof of \textsc{$\A$-Contains-Sub}]
Let $x \in \A(p_1, exp_1, b_1)$.
Since $x$ has quantum at least $2^{exp_1}$,
  write $x = (-1)^s \times c \times 2^{exp_1}$
  for some $c \geq 0$.
From $|x| \leq b_1 \leq 2^{exp_1 + p_2}$
  we obtain $c \times 2^{exp_1} = |x| \leq 2^{exp_1 + p_2}$,
  so $c \leq 2^{p_2}$.
An integer of magnitude at most $2^{p_2}$ has
  a significand of at most $p_2$ digits,
  so $x$ has precision at most
  $p_2$---even when $p_1 > p_2$.
Moreover, since $exp_2 \leq exp_1$, the value
\[
  x = (-1)^s \times (c \times 2^{exp_1 - exp_2}) \times 2^{exp_2}
\]
is an integer multiple of $2^{exp_2}$,
  so it satisfies the quantum constraint.
Finally, $|x| \leq b_1 \leq b_2$,
  and we conclude $x \in \A(p_2, exp_2, b_2)$.
\end{proof}

\begin{proof}[Proof of \refLemContainComplete]
We need to prove that if $F_1 \subseteq F_2$,
  then \textsc{$\A$-Contains-Prec} or \textsc{$\A$-Contains-Sub} applies.
Suppose $F_1 \subseteq F_2$,
  and let $x_0 \neq 0$ be a value of $F_1$,
  which exists since $F_1$ is non-trivial.

\emph{Bound.}
If $b_2 = \infty$, then $b_1 \leq b_2$.
Otherwise, if $b_1$ is finite, then $b_1 \in F_1$
  since $b_1$ is representable in $\A(p_1, exp_1, \infty)$,
  so $b_1 \in F_2$ and $b_1 \leq b_2$.
If $b_1 = \infty$, then $F_1$ contains $2^k$
  for every $k \geq exp_1$, including some $2^k > b_2$,
  which is not in $F_2$; a contradiction.

\emph{Exponent.}
If $exp_2 = -\infty$, then $exp_1 \geq exp_2$.
Otherwise, if $exp_1$ is finite, then
  $2^{exp_1} \leq |x_0| \leq b_1$, so $2^{exp_1} \in F_1 \subseteq F_2$;
  since every value of $F_2$ is a multiple of $2^{exp_2}$,
  $exp_1 \geq exp_2$.
If $exp_1 = -\infty$, then $F_1$ contains $2^k$ for some $k < exp_2$
  with $2^k \leq |x_0|$, which is not a multiple of $2^{exp_2}$;
  a contradiction.

\emph{Precision.}
If $p_1 \leq p_2$, then \textsc{$\A$-Contains-Prec} applies.
Otherwise $p_2 < p_1$, so $p_2$ is finite.
For any $k$, the value $w_k = (2^{p_2} + 1) \times 2^k$
  has an odd significand of $p_2 + 1 \leq p_1$ digits,
  so $w_k \notin F_2$, and hence $w_k \notin F_1$:
  either $k < exp_1$ or $w_k > b_1$.
If $exp_1 = -\infty$, some $w_k \leq |x_0| \leq b_1$ would lie in $F_1$;
  so $exp_1$ is finite.
Finally, suppose $b_1 > 2^{exp_1 + p_2}$.
If $b_1 = \infty$, then $w_{exp_1} \in F_1$ directly;
  otherwise $b_1 = m \times 2^{exp_1}$ for an integer $m > 2^{p_2}$,
  so $w_{exp_1} \leq b_1$ and again $w_{exp_1} \in F_1$.
Either way we have a contradiction,
  so $b_1 \leq 2^{exp_1 + p_2}$.
With the bound and exponent conditions,
  \textsc{$\A$-Contains-Sub} applies.
\end{proof}

Both restrictions on $F_1$ are necessary.
Without a representable bound,
  the formats $\A(\infty, 0, 3/2)$ and $\A(\infty, 0, 1)$ are both $\{0, \pm 1\}$,
  yet neither rule applies, since $3/2 > 1$.
Without non-triviality,
  $\A(\infty, 0, 0) = \A(\infty, 1, 0) = \{0\}$,
  yet neither rule applies, since $0 < 1$.

\subsection{Correct Double Rounding}
\label{ssec:appendix:proof:double_round}

We prove soundness of the rules in \refFigCorrectDoubleRound.
The first six rules establish
  $\langle F_1, rm_1 \rangle \rhd \langle F_2, rm_2 \rangle$
  for rounding modes $rm_1, rm_2$ among RNE, RTZ, RAZ, and RTO:
  we first show that rounding a real number through $F_2$
  matches rounding it directly to $F_1$ when neither overflows,
  and then that the two overflow together.
\textsc{rnd-Double} then extends the result to special values.
Let $x \in \R$ be a real number.
We first show that
  the number of digits that $\rnd_{\A(p_1, exp_1, b_1), rm}(x)$
  rounds to is a function of $p_1, exp_1$, and $x$:


\begin{proof}[Proof of \refLemKeptDigits]
If $x = 0$,
  then the number of digits rounded to is $0$,
  so assume $x \neq 0$ and
  let $e = \lfloor \log_2 |x| \rfloor$ denote the exponent of $x$.
We compute $exp'$, the position of the least significant digit
  that $\rnd_{\A(p, exp, b), rm}(x)$ retains,
  by case analysis on $p$ and $exp$:
\begin{itemize}
  \item If $p = \infty$ and $exp > -\infty$,
    then $\A(p, exp, b)$ is a fixed-point format with quantum $2^{exp}$,
    so $exp' = exp$.
  \item If $p < \infty$ and $exp = -\infty$,
    then $\A(p, exp, b)$ is a floating-point format with precision $p$,
    so $exp' = e - p + 1$.
  \item If $p < \infty$ and $exp > -\infty$,
    then $\A(p, exp, b)$ is a floating-point format with precision $p$
    and minimum quantum $2^{exp}$,
    so $exp' = \max(e - p + 1, exp)$.
\end{itemize}
In every case $exp'$ depends only on $p$, $exp$, and $x$;
  in particular, it is independent of the rounding mode $rm$
  and of the bound $b$.
Hence the number of digits retained,
  $w = e - exp' + 1$,
  is a well-defined function of $p$, $exp$, and $x$ alone.
\end{proof}

Let $\A(p_1, exp_1, b_1)$ be a number format and let $k \geq 0$.
Write $w_1$ and $w_2$ for the number of digits that
  $\rnd_{\A(p_1, exp_1, b_1), rm}(x)$ and
  $\rnd_{\A(p_1 + k, exp_1 - k, b_1), rm}(x)$ round to.
We show $w_2 = w_1 + k$ for every $x \neq 0$.

\begin{proof}[Proof of \refLemKeptCmp]
Fix $x \neq 0$,
  and let $e = \lfloor \log_2 |x| \rfloor$ denote
  the exponent of $x$.
By \refLemKeptDigits,
  we case split on $p_1$ and $exp_1$
  to determine $exp'_1$, the position of
  the least significant digit that rounding
  to $\A(p_1, exp_1, \infty)$ retains:
\begin{itemize}
  \item If $p_1 = \infty$ and $exp_1 > -\infty$,
    then $exp'_1 = exp_1$.
    Since $p_1 + k = \infty$ and $exp_1 - k > -\infty$,
    rounding to $\A(p_1 + k, exp_1 - k, \infty)$ gives
    $exp'_2 = exp_1 - k = exp'_1 - k$.
  \item If $p_1 < \infty$ and $exp_1 = -\infty$,
    then $exp'_1 = e - p_1 + 1$.
    Since $p_1 + k < \infty$ and $exp_1 - k = -\infty$,
    rounding to $\A(p_1 + k, exp_1 - k, \infty)$ gives
    $exp'_2 = e - (p_1 + k) + 1 = exp'_1 - k$.
  \item If $p_1 < \infty$ and $exp_1 > -\infty$,
    then $exp'_1 = \max(e - p_1 + 1, exp_1)$.
    Since $p_1 + k < \infty$ and $exp_1 - k > -\infty$,
    rounding to $\A(p_1 + k, exp_1 - k, \infty)$ gives
    \[
      exp'_2 = \max(e - (p_1 + k) + 1,\, exp_1 - k)
             = \max(e - p_1 + 1,\, exp_1) - k
             = exp'_1 - k.
    \]
\end{itemize}
In all cases $exp'_2 = exp'_1 - k$,
  so $w_2 = e - exp'_2 + 1 = (e - exp'_1 + 1) + k = w_1 + k$.
\end{proof}

In the case that $k \geq 1$,
  rounding under RTO has a special property:




\begin{proof}[Proof of \refLemRtoRounding]
We prove the biconditional, then the interval bound.

\emph{($\Rightarrow$)} If $x$ is representable at precision $w$,
  it is also representable at precision $w + k$
  (pad with $k$ trailing zero digits),
  so RTO leaves it unchanged: $x' = x$.
  Hence $x'$ is representable at precision $w$.

\emph{($\Leftarrow$)} We prove the contrapositive:
  if $x$ is unrepresentable at precision $w$, then so is $x'$.
  If $x$ is representable at precision $w + k$,
    then $x' = x$, which is unrepresentable at precision $w$ by assumption.
  Otherwise $x$ is unrepresentable at precision $w + k$,
    and RTO rounds it to a value with an odd significand.
  Since $k \geq 1$, the $(w + k)$-th digit of $x'$ is nonzero
    and lies below the $w$-th digit,
    so $x'$ is unrepresentable at precision $w$.

For the interval bound, suppose $x$ is unrepresentable at precision $w$
  and let $y_1 < x < y_2$ be its two nearest representable neighbors
  at precision $w$.
If $x' = x$, then $y_1 < x' < y_2$ is immediate.
Otherwise $x$ is unrepresentable at precision $w + k$,
  so $x'$ has an odd significand at precision $w + k$,
  whereas $y_1$ and $y_2$---%
  being representable at precision $w$---%
  have even significands at precision $w + k$
  (their low $k$ digits are zero).
RTO never rounds to an even value,
  so $x' \notin \{y_1, y_2\}$;
  and since RTO rounds $x$ to one of
  its two neighbors at precision $w + k$,
  both of which lie in $[y_1, y_2]$,
  we conclude $y_1 < x' < y_2$.
\end{proof}

We now prove each of the correct double rounding rules.
We'll start with the first three rules where
  both roundings use the same rounding mode.






\begin{proof}[Proof of \textsc{rnd-RTZ-RTZ}]
Write $F_1 = \A(p_1, exp_1, b_1)$ and $F_2 = \A(p_2, exp_2, b_2)$.
Without loss of generality assume $x > 0$; the case $x < 0$ is symmetric,
  and $x = 0$ is fixed by RTZ in both formats.
For a positive argument, RTZ returns the largest representable value of the
  format that does not exceed the argument.

Let $z = \rnd_{F_2, \text{RTZ}}(x)$ and $w = \rnd_{F_1, \text{RTZ}}(z)$,
  so $w \leq z \leq x$ and $w \in F_1$.
It remains to show that $w$ is the largest element of $F_1$ that does not
  exceed $x$.
Let $y \in F_1$ with $y \leq x$.
Since $F_1 \subseteq F_2$, we also have $y \in F_2$,
  so $y \leq z$, as $z$ is the largest element of $F_2$ not exceeding $x$.
Then $y \leq w$, as $w$ is the largest element of $F_1$ not exceeding $z$.
Hence $w$ is the largest element of $F_1$ not exceeding $x$,
  that is, $w = \rnd_{F_1, \text{RTZ}}(x)$.

\end{proof}






\begin{proof}[Proof of \textsc{rnd-RAZ-RAZ}]
Write $F_1 = \A(p_1, exp_1, b_1)$ and $F_2 = \A(p_2, exp_2, b_2)$.
Without loss of generality assume $x > 0$; the case $x < 0$ is symmetric,
  and $x = 0$ is fixed by RAZ in both formats.
For a positive argument, RAZ returns the smallest representable value of the
  format that is at least the argument.

Let $z = \rnd_{F_2, \text{RAZ}}(x)$ and $w = \rnd_{F_1, \text{RAZ}}(z)$,
  so $w \geq z \geq x$ and $w \in F_1$.
It remains to show that $w$ is the smallest element of $F_1$ that is at least
  $x$.
Let $y \in F_1$ with $y \geq x$.
Since $F_1 \subseteq F_2$, we also have $y \in F_2$,
  so $y \geq z$, as $z$ is the smallest element of $F_2$ at least $x$.
Then $y \geq w$, as $w$ is the smallest element of $F_1$ at least $z$.
Hence $w$ is the smallest element of $F_1$ at least $x$,
  that is, $w = \rnd_{F_1, \text{RAZ}}(x)$.

\end{proof}




\begin{proof}[Proof of \textsc{rnd-RTO-RTO}]
Write $F_1 = \A(p_1, exp_1, b_1)$ and $F_2 = \A(p_2, exp_2, b_2)$,
  where $p_2 \geq 2$.
If $x \in F_1$, both roundings fix $x$ and the identity is immediate, so assume
  $x \notin F_1$.

Let $w_1$ and $w_2$ be the number of digits rounded to when rounding $x$ into
  $F_1$ and $F_2$, respectively; since $F_1 \subseteq F_2$, $w_1 \leq w_2$.
Let $z = \rnd_{F_2, \text{RTO}}(x)$ and $w = \rnd_{F_1, \text{RTO}}(z)$.
If $w_1 = w_2$, then $F_2$ rounds $x$ at the same precision as $F_1$, so
  $z = \rnd_{F_1, \text{RTO}}(x) \in F_1$ and $w = z$.
Otherwise $w_1 + 1 \leq w_2$, and by \refLemRtoRounding, $z$ is unrepresentable
  in $F_1$ and bracketed by the same two elements of $F_1$ as $x$; rounding $z$
  in $F_1$ therefore yields the same odd value as rounding $x$.
In either case,
  $w = \rnd_{F_1, \text{RTO}}(x)$.

\end{proof}

Finally,
  we prove the last three rules where the first rounding uses RTO
  and the second rounding uses a different rounding mode.




\begin{proof}[Proof of \textsc{rnd-RTO-RTZ}]
Write $F_1 = \A(p_1, exp_1, b_1)$ and $F_2 = \A(p_2, exp_2, b_2)$.
Without loss of generality assume $x > 0$; the case $x < 0$ is symmetric, and
  $x = 0$ is fixed by both roundings.
If $x \in F_1$, the identity is immediate, so assume $x \notin F_1$.

Let $w_1$ and $w_2$ be the number of digits rounded to when rounding $x$ into
  $F_1$ and $F_2$.
Since $\A(p_1 + 1, exp_1 - 1, b_1) \subseteq F_2$, we have $w_1 + 1 \leq w_2$
  by \refLemKeptCmp.
Let $y_1 < x < y_2$ be the adjacent elements of $F_1$ about $x$, so that
  $\rnd_{F_1, \text{RTZ}}(x) = y_1$.
Set $z = \rnd_{F_2, \text{RTO}}(x)$ and $w = \rnd_{F_1, \text{RTZ}}(z)$.
Since $w_1 + 1 \leq w_2$, \refLemRtoRounding applies: $z$ is unrepresentable in
  $F_1$ and $y_1 < z < y_2$.
Hence $z$ shares the $F_1$-neighbors $y_1, y_2$, and rounding it toward zero
  gives $w = y_1 = \rnd_{F_1, \text{RTZ}}(x)$.

\end{proof}




\begin{proof}[Proof of \textsc{rnd-RTO-RAZ}]
Write $F_1 = \A(p_1, exp_1, b_1)$ and $F_2 = \A(p_2, exp_2, b_2)$.
Without loss of generality assume $x > 0$; the case $x < 0$ is symmetric, and
  $x = 0$ is fixed by both roundings.
If $x \in F_1$, the identity is immediate, so assume $x \notin F_1$.

Let $w_1$ and $w_2$ be the number of digits rounded to when rounding $x$ into
  $F_1$ and $F_2$.
Since $\A(p_1 + 1, exp_1 - 1, \text{next}_{p_1, exp_1}(b_1)) \subseteq F_2$,
  we have $w_1 + 1 \leq w_2$ by \refLemKeptCmp.
Let $y_1 < x < y_2$ be the adjacent elements of $F_1$ about $x$, so that
  $\rnd_{F_1, \text{RAZ}}(x) = y_2$.
Set $z = \rnd_{F_2, \text{RTO}}(x)$ and $w = \rnd_{F_1, \text{RAZ}}(z)$.
Since $w_1 + 1 \leq w_2$, \refLemRtoRounding applies: $z$ is unrepresentable in
  $F_1$ and $y_1 < z < y_2$.
Hence $z$ shares the $F_1$-neighbors $y_1, y_2$, and rounding it away from zero
  gives $w = y_2 = \rnd_{F_1, \text{RAZ}}(x)$.

\end{proof}




\begin{proof}[Proof of \textsc{rnd-RTO-RNE}]
Write $F_1 = \A(p_1, exp_1, b_1)$ and $F_2 = \A(p_2, exp_2, b_2)$.
Without loss of generality assume $x > 0$; the case $x < 0$ is symmetric, and
  $x = 0$ is fixed by both roundings.
If $x \in F_1$, the identity is immediate, so assume $x \notin F_1$.

Let $w_1$ and $w_2$ be the number of digits rounded to when rounding $x$ into
  $F_1$ and $F_2$.
Since $\A(p_1 + 2, exp_1 - 2, b_1) \subseteq F_2$, we have $w_1 + 2 \leq w_2$
  by \refLemKeptCmp.
Let $y_1 < x < y_2$ be the adjacent elements of $F_1$ about $x$, with midpoint
  $m = (y_1 + y_2)/2$; then $\rnd_{F_1, \text{RNE}}(x)$ is the nearer of
  $y_1, y_2$, or the even endpoint when $x = m$.
Set $z = \rnd_{F_2, \text{RTO}}(x)$ and $w = \rnd_{F_1, \text{RNE}}(z)$.
Since $w_1 + 2 \leq w_2$, \refLemRtoRounding applies: $z$ is unrepresentable in
  $F_1$ and $y_1 < z < y_2$.
The midpoint $m$ is representable at precision $w_1 + 1$, and $w_2 \geq w_1 + 2$,
  so RTO resolves $x$'s position relative to $m$: $z = m$ exactly when $x = m$,
  and otherwise the extra digit below $m$ keeps $z$ strictly on the same side of
  $m$ as $x$.
In either case,
  $z$ and $x$ round to the same $F_1$-value under RNE---the nearer
  endpoint when $x \neq m$, the even endpoint when $x = m$---so
  $w = \rnd_{F_1, \text{RNE}}(z) = \rnd_{F_1, \text{RNE}}(x)$.

\end{proof}

\subsubsection*{Overflow}
The proofs above assume no rounding overflows.
We now show that both roundings overflow together,
  completing $\langle F_1, rm_1 \rangle \rhd \langle F_2, rm_2 \rangle$.
Fix $x \in \R$ and let $y$, $z$, and $w$ be the roundings of
  $x$ into $F_1$, of $x$ into $F_2$, and of $z$ into $F_1$,
  each ignoring the bound;
  rounding directly overflows when $|y| > b_1$,
  and rounding through $F_2$ when $|z| > b_2$ or $|w| > b_1$.
Each rule has a threshold $T$, a value of $F_2$ by its containment premise,
  with two properties:
\begin{itemize}
\item if $|x| > T$, then $|y| > b_1$ and $|w| > b_1$, so both overflow;
\item if $|x| \leq T$, then $|z| \leq T \leq b_2$,
  and the rule's proof applies to $x$,
  since the premise places the values of $F_1$ up to $T$ in $F_2$;
  so $w = y$, and both overflow or neither does.
\end{itemize}
Both use that rounding is monotone and returns representable values unchanged,
  so $|x| \leq T$ implies $|z| \leq T$ for $T \in F_2$,
  and $|x| \geq T$ implies $|y| \geq T$ for $T \in \A(p_1, exp_1, \infty)$.
For \textsc{rnd-RAZ-RAZ}, $T = b_1$:
  RAZ never decreases magnitude,
  so $|x| > b_1$ gives $|y| \geq |x| > b_1$ and $|w| \geq |z| \geq |x| > b_1$.
For \textsc{rnd-RTZ-RTZ}, \textsc{rnd-RTO-RTO}, \textsc{rnd-RTO-RTZ}, and \textsc{rnd-RTO-RAZ},
  $T = \text{next}_{p_1, exp_1}(b_1)$, a value of $\A(p_1, exp_1, \infty)$ above $b_1$:
  $|x| > T$ gives $|y| \geq T$ and $|z| \geq T$, hence $|w| \geq T > b_1$.
For \textsc{rnd-RTO-RNE},
  $T = \text{next}_{p_1 + 1, exp_1 - 1}(b_1)$,
  the midpoint between $b_1$ and $\text{next}_{p_1, exp_1}(b_1)$:
  $|x| > T$ gives $|y| \geq \text{next}_{p_1, exp_1}(b_1)$;
  and $|z| > T$, since RTO never rounds an inexact value onto $T$,
  which has fewer digits than $F_2$ keeps,
  so $|w| \geq \text{next}_{p_1, exp_1}(b_1)$ as well.

\begin{proof}[Proof of \textsc{rnd-Double}]
Let $v$ be an input.
If $v$ is a special value,
  rounding through $F_2$ gives
  $\rnd^{\sigma_1, \omega_1}_{F_1, rm_1}(\sigma_2(v)) = \sigma_1(v)$,
  the direct rounding.
Otherwise $v = x \in \R$,
  and by $\langle F_1, rm_1 \rangle \rhd \langle F_2, rm_2 \rangle$
  either neither rounding overflows and the results agree,
  or both overflow and the direct rounding gives $\omega_1(\mathrm{sgn}(x))$.
In the latter case, rounding preserves sign:
  if $|z| > b_2$, rounding through $F_2$ gives
  $\rnd^{\sigma_1, \omega_1}_{F_1, rm_1}(\omega_2(\mathrm{sgn}(x))) = \omega_1(\mathrm{sgn}(x))$;
  otherwise $|w| > b_1$ and it gives $\omega_1(\mathrm{sgn}(z)) = \omega_1(\mathrm{sgn}(x))$.
\end{proof}

\subsection{Format Inference}
\label{ssec:appendix:proof:infer}

Recall that $x \in \A(p, exp, b)$ means $x$ has precision at most $p$
  (a significand of magnitude below $2^p$), is an integer multiple of
  $2^{exp}$, and satisfies $|x| \leq b$.
When an input parameter is infinite ($p = \infty$, $exp = -\infty$, or
  $b = \infty$), the matching output parameter is infinite too and the
  corresponding constraint is vacuous; we treat only the finite cases.
Finite operands have finite sums and products,
  so special values arise only from special operands;
  the sets $S_\oplus$ and $S_\otimes$ follow by cases
  on the IEEE 754 rules,
  and we prove the numerical part below.

\begin{proof}[Proof of the $\otimes$ rule]
Let $x \in \A(p_1, exp_1, b_1)$ and $y \in \A(p_2, exp_2, b_2)$.
We show $x \cdot y \in \A(p_1 + p_2,\, exp_1 + exp_2,\, b_1 \times b_2)$.

\emph{Precision.}
The precision constraints give representations
  $x = s_1 \times 2^{e_1}$ and $y = s_2 \times 2^{e_2}$
  with $|s_1| < 2^{p_1}$ and $|s_2| < 2^{p_2}$.
Then
\[
  x \cdot y = (s_1 s_2) \times 2^{e_1 + e_2},
  \qquad
  |s_1 s_2| = |s_1|\,|s_2| < 2^{p_1} \cdot 2^{p_2} = 2^{p_1 + p_2},
\]
so $x \cdot y$ has precision at most $p_1 + p_2$.

\emph{Quantum.}
The quantum constraints give $x = c_1 \times 2^{exp_1}$ and
  $y = c_2 \times 2^{exp_2}$ for integers $c_1, c_2$.
Then $x \cdot y = (c_1 c_2) \times 2^{exp_1 + exp_2}$ is an integer multiple
  of $2^{exp_1 + exp_2}$.

\emph{Bound.}
From $|x| \leq b_1$ and $|y| \leq b_2$ we get
  $|x \cdot y| = |x|\,|y| \leq b_1 b_2$.
\end{proof}

\begin{proof}[Proof of the $\oplus$ rule]
Let $x \in \A(p_1, exp_1, b_1)$ and $y \in \A(p_2, exp_2, b_2)$.
We show $x + y \in \A(P,\, m,\, b_1 + b_2)$, where
\[
  m = \text{min}(exp_1, exp_2),
  \qquad
  N = \lfloor (b_1 + b_2)/2^{m} \rfloor,
  \qquad
  P = \max\!\big(1, \lceil \log_2(N + 1) \rceil\big).
\]

\emph{Quantum.}
Since $exp_1, exp_2 \geq m$, the quantum constraints
  $x = c_1 \times 2^{exp_1}$ and $y = c_2 \times 2^{exp_2}$
  rewrite at the finer quantum $2^m$ as
\[
  x = (c_1 \times 2^{exp_1 - m}) \times 2^m,
  \qquad
  y = (c_2 \times 2^{exp_2 - m}) \times 2^m,
\]
with integer coefficients (the exponents $exp_i - m$ are nonnegative).
Hence $x + y = c \times 2^m$ where
  $c = c_1 \times 2^{exp_1 - m} + c_2 \times 2^{exp_2 - m}$
  is an integer, so $x + y$ is an integer multiple of $2^m$.

\emph{Bound.}
$|x + y| \leq |x| + |y| \leq b_1 + b_2$.

\emph{Precision.}
Reusing the representation $x + y = c \times 2^m$ from the quantum step,
  the bound step gives
\[
  |c| \times 2^m = |x + y| \leq b_1 + b_2,
  \qquad\text{so}\qquad
  |c| \leq (b_1 + b_2)/2^m.
\]
As $c$ is an integer, in fact $|c| \leq \lfloor (b_1 + b_2)/2^m \rfloor = N$.
By definition of the ceiling logarithm, $N + 1 \leq 2^{\lceil \log_2(N+1) \rceil}
  \leq 2^{P}$, so
\[
  |c| \leq N < N + 1 \leq 2^{P}.
\]
Thus $x + y = c \times 2^m$ with $|c| < 2^{P}$; that is, $x + y$ has precision
  at most $P$.
\end{proof}

\fi

\end{document}